\documentclass[twocolumn,preprintnumbers,amsmath,amssymb]{revtex4-2}
\usepackage[english]{babel}
\usepackage{enumitem}
\usepackage{epsfig}
\usepackage{graphicx}
\usepackage{dcolumn}
\usepackage{bm}
\usepackage{amsmath,amsfonts}
\usepackage{amssymb}
\usepackage{amsthm}
\usepackage{physics}
\usepackage{comment}
\usepackage{color} 
\usepackage{xcolor}
\usepackage[final,colorlinks,linkcolor=blue,anchorcolor=blue,urlcolor=blue,citecolor=blue]{hyperref}
\usepackage{caption}
\usepackage{subcaption}
\usepackage{booktabs}
\usepackage{siunitx}
\usepackage{epstopdf}
\usepackage{diagbox}
\usepackage{framed}
\usepackage{dsfont}
\usepackage{float}
\newtheorem{theorem}{Theorem}
\newtheorem{proposition}{Proposition}
\newtheorem{corollary}{Corollary}

\begin{document}
	\title{Sharp Finite Statistics for Quantum Key Distribution}

\author{Vaisakh Mannalath$^{1,2,3}$  }
\thanks{These authors contributed equally to this Letter. \\ vmannalath@vqcc.uvigo.es \\ vzapatero@vqcc.uvigo.es}	
         
	\author{Víctor Zapatero$^{1,2,3}$}
           \thanks{These authors contributed equally to this Letter. \\ vmannalath@vqcc.uvigo.es \\ vzapatero@vqcc.uvigo.es}
           
        \author{Marcos Curty$^{1,2,3}$}
        \affiliation{$^1$Vigo Quantum Communication Center, University of Vigo, Vigo E-36310, Spain}
        \affiliation{$^2$Escuela de Ingeniería de Telecomunicación, Department of Signal Theory and Communications, University of Vigo, Vigo E-36310, Spain}
        \affiliation{$^3$AtlanTTic Research Center, University of Vigo, Vigo E-36310, Spain}        

\begin{abstract}
The performance of quantum key distribution (QKD) heavily depends on statistical inference. For a broad class of protocols, the central statistical task is a random sampling problem, customarily addressed using a hypergeometric tail bound due to Serfling. Here, we provide an alternative solution for this task of unprecedented tightness among QKD security analyses. As a by-product, confidence intervals for the average of nonidentical Bernoulli parameters follow too. These naturally fit in statistical analyses of decoy-state QKD and also outperform standard tools. Last, we show that, in a vast parameter regime, the use of tail bounds is not enforced because the cumulative mass function of the hypergeometric distribution is accurately computable. This sharply decreases the minimum block sizes necessary for QKD, and reveals the tightness of our analytical bounds when moderate-to-large blocks are considered.
\end{abstract}
\maketitle
\textit{Introduction}---Quantum key distribution (QKD) is the only approach to key exchange that is provably secure against computationally unbounded adversaries~\cite{PortmannRenner}. Today, considerable efforts are being made to boost the performance of practical QKD systems while guaranteeing their physical-layer security~\cite{Feihu,review}. Among the multiple facets of this ongoing program, it is of utmost importance to achieve tighter lower bounds on the extractable secret key length, a goal that can be tackled by refining QKD security proofs. These proofs build upon quantum information-theoretic results and strongly rely on statistical inference, which plays a pivotal role to keep the necessary data block sizes to a minimum.

To be precise, establishing the finite-key security of a QKD protocol requires to upper bound the probability of a ``failure" happening in its parameter estimation (PE) step. This refers to the event where the protocol does not abort but the adversary obtains nontrivial information about the final key~\cite{PortmannRenner}. In BB84-like protocols~\cite{BB84}, a very popular proof technique is based on an entropic uncertainty relation (EUR) for smooth entropies~\cite{EUR,EUR_review}, in which case the simplest formulation of the problem goes as follows. Let us consider a two-piece experiment where a population of $N$ binary random variables is generated according to an arbitrary probability distribution, and a random sample of size $n$ is drawn from the population afterward. If we denote the frequency of ones in the test (complementary) sample as $\hat{p}$ ($\hat{q}$), the goal is to find a threshold value $q^{\rm th}\geq{}p^{\rm th}$ such that
\begin{equation}\label{goal}
\Pr\Bigl[\hat{p}\leq{}p^{\rm th},\hat{q}\geq{}q^{\rm th}\Bigr]\leq{}\epsilon,
\end{equation}
for arbitrary values of $N$, $n$, $\epsilon$ and $p^{\rm th}$. Naturally, an interplay exists between $\epsilon$ and $q^{\rm th}$: fulfilling Eq.~(\ref{goal}) with a very small $\epsilon$ comes at the price of tolerating larger values of $q^{\rm th}$. In the context of QKD---in which $\hat{q}$ is often referred to as the phase-error rate---this translates into a direct penalty on the provably secure key length. Notably, the above formulation of the problem shows up exactly in ideal implementations of the BB84~\cite{BB84} or the BBM92~\cite{BBM92,Leverrier} protocols, and we adhere to the latter in the main text for didactic purposes. The detailed protocol description and secret key length formula can be found in Appendix I. On the other hand, the popular decoy-state BB84 protocol~\cite{decoy2,decoy,decoy3,Lim}, which is implemented in most of today's commercial QKD setups, is addressed in Appendix II.

In EUR-based security proofs~\cite{Leverrier,Tomamichel}, it is commonplace to derive the desired $q^{\rm th}$ by resorting to an exponential tail bound on the hypergeometric (HG) distribution known as Serfling inequality~\cite{Serfling}. Recently though, a sharper bound has been derived in~\cite{Ekert}, combining Serfling inequality with a refined inequality due to Hush and Scovel~\cite{Hush}. Nonetheless, the reported improvement is moderate and heavily relies on the use of very large samples for testing, an undesirable feature given the fact that the final key is extracted from the complementary data.

In this Letter, we propose a simpler approach to calculate $q^{\rm th}$, resorting to a well-known link between random sampling with and without replacement~\cite{Hoeffding}. Remarkably, the bound we obtain outperforms~\cite{Ekert} in all the explored parameter regimes relevant for QKD, where one restricts $p^{\rm th}$ to a few percent at most, $N$ to thousands or tens of thousands at least, and $\epsilon$ to rather minuscule values (often below $10^{-10}$).
Furthermore, in the derivation of the bound, confidence intervals (CIs) are obtained for the average parameter of a set of independent Bernoulli trials. These intervals, which provide an additional tool required for current decoy-state QKD security proofs, also outperform the analytical results considered for that purpose so far~\cite{Lim,Zhang}. Last, we show that, contrary to general belief, the cumulative mass function (CMF) of the HG distribution is exactly computable in a vast parameter regime of interest to QKD. Importantly, this enables the calculation of optimal~\cite{Buonaccorsi,Wright,Weizhen} 
confidence upper bounds for the population parameter following the Clopper-Pearson method~\cite{Clopper,optimality}. Applying the latter to the random sampling problem of Eq.~(\ref{goal}) ---and also to its generalization to decoy-state QKD schemes---we report a dramatic decrease of the minimum block sizes necessary for key extraction. On the other hand, explicit comparison with our simple analytical tools reveals the tightness of the latter for moderate-to-large blocks.\\

\textit{Analytical results}---Our analytical results are triggered by two observations. First, the additive Chernoff bound~\cite{Chernoff} (simply Chernoff bound in what follows), conventionally used in the context of sampling with replacement, also holds for sampling without replacement. This feature, which is a consequence of a general theorem by Hoeffding~\cite{Hoeffding}---see also~\cite{Chvatal} for a simple proof--- has already been exploited in QKD security analyses before~\cite{Hayashi_1,Hayashi_2}. Secondly, Serfling inequality---the standard bound used in sampling without replacement---essentially reduces to Hoeffding inequality~\cite{Hoeffding} for small sampling fractions $n/N$. Nevertheless, the latter constitutes a loose relaxation of the Chernoff bound, revealing that this bound must certainly be better for small $n/N$. With this in mind, we derive a tight relaxation of the Chernoff bound and indistinctly use it to obtain one-sided CIs for the population parameters in sampling with and without replacement. Then, we apply these CIs to solve Eq.~(\ref{goal}) and its generalization to decoy-state QKD.

Crucially, our relaxation stems from a rational approximation of the Kullback-Leibler divergence, which is based on a logarithmic inequality~\cite{Topsoe} carefully selected to keep the CIs tight and handy at the same time. Indeed, higher-order rational approximations of the logarithm lead to much more convoluted formulas for the divergence. The reader is referred to Appendix III for the relaxed Chernoff bound, and the corresponding CIs are stated in Proposition 1 and Corollary 1 below.

\begin{proposition} Let $\hat{p}$ be the average of $n$ independent Bernoulli variables, with expected value $\mathbb{E}\left[\hat{p}\right]=p$, and let $\kappa_{n,\epsilon}=(2/9n)\ln(1/\epsilon)$. Then, for all $\epsilon>0$, $\Pr[\Gamma^-_{n,\epsilon}(\hat{p})\geq{}p]\leq{\epsilon}$ and $\Pr[\Gamma^+_{n,\epsilon}(\hat{p})\leq{}p]\leq{\epsilon}$, where
\begin{equation}\label{Gamma-}
\Gamma^-_{n,\epsilon}(\hat{p})=\left\{
\begin{array}{ll}
\gamma^-_{n,\epsilon}(\hat{p}) & \mathrm{if}\hspace{.2cm}\hat{p}\in{}\left[\displaystyle{\frac{3\kappa_{n,\epsilon}}{1+\kappa_{n,\epsilon}}},1\right], \\
-\epsilon & \mathrm{otherwise},
\end{array} 
\right.
\end{equation}
\begin{equation}\label{Gamma+}
\Gamma^+_{n,\epsilon}(\hat{p})=\left\{
\begin{array}{ll}
\gamma^+_{n,\epsilon}(\hat{p}) & \mathrm{if}\hspace{.2cm}\hat{p}\in\left[0,\displaystyle{\frac{1-2\kappa_{n,\epsilon}}{1+\kappa_{n,\epsilon}}}\right], \\
1+\epsilon & \mathrm{otherwise},
\end{array} 
\right.
\end{equation}
and
\begin{eqnarray}\label{U}
\gamma^\pm_{n,\epsilon}(x)&&=\frac{1}{1+4\kappa_{n,\epsilon}}\biggl[3\kappa_{n,\epsilon}+(1-2\kappa_{n,\epsilon})x\nonumber \\
&&\pm3\sqrt{\kappa_{n,\epsilon}\left(\kappa_{n,\epsilon}+x-x^{2}\right)}\biggr].
\end{eqnarray}
\end{proposition}
A proof of this claim is given in Appendix III. Also, the simpler claims
$\Pr\left[\max\{0,\gamma^{-}_{n,\epsilon}(\hat{p})\}>p\right]\leq{}\epsilon$ and $\Pr\left[\min\{1,\gamma^{+}_{n,\epsilon}(\hat{p})\}<p\right]\leq{}\epsilon$ follow identically~\cite{justification}.

Since Proposition 1 is solely derived from the Chernoff bound, it automatically extends to a HG random variable~\cite{Hoeffding,Chvatal}. The upper extension goes as follows.

\begin{corollary}
Let $\hat{X}\sim{}Hypergeometric(N,K,n)$ and $\hat{p}=\hat{X}/n$, such that $\mathbb{E}\left[\hat{p}\right]=p$ with $p=K/N$. Then, for all $\epsilon>0$, $\Pr[\Gamma^+_{n,\epsilon}(\hat{p})\leq{}p]\leq{\epsilon}$, where $\Gamma^+_{n,\epsilon}(x)$ is defined in Proposition 1.
\end{corollary}

This result readily implies the following proposition.

\begin{proposition} A test sample of size $n$ is drawn at random from a binary population with $N$ elements. Let $\hat{p}$ ($\hat{q}$) denote the frequency of ones in the test (complementary) sample. For all $\epsilon>0$,
\begin{equation}
\Pr\left[\hat{q}\geq{}q^{\rm th}_{N,n,\epsilon}(\hat{p})\right]\leq{}\epsilon
\end{equation}
for
\begin{equation}\label{q_th}
q^{\rm th}_{N,n,\epsilon}(x)=\frac{N\Gamma^+_{n,\epsilon}(x)-nx}{N-n},
\end{equation}
where $\Gamma^+_{n,\epsilon}(x)$ is defined in Proposition 1.
\end{proposition}

\begin{proof}
Let $\hat{X}$ be the number of ones in the test sample, such that $\hat{p}=\hat{X}/n$. If we denote by $p$ the fraction of ones in the population, $\hat{X}\sim{}\mathrm{Hypergeometric}(N,Np,n)$ and the claim follows directly from Corollary 1 by substituting $p=n\hat{p}/N+(N-n)\hat{q}/N$ and rearranging terms.
\end{proof}

As shown next, Proposition 2 allows to obtain a threshold value $q^{\rm th}$ fulfilling Eq.~(\ref{goal}), which provides the so-called failure probability estimation of the BBM92 protocol.

\begin{proposition} Consider the two-piece experiment presented in the Introduction. Then, for all $\epsilon>0$ and $p^{\rm th}<1/2$,
\begin{equation}\label{upperbound_q}
\Pr\Bigl[\hat{p}\leq{}p^{\rm th},\hat{q}\geq{}q^{\rm th}_{N,n,\epsilon}\left(p^{\rm th}\right)\Bigr]\leq{}\epsilon
\end{equation}
for the threshold function $q^{\rm th}_{N,n,\epsilon}(x)$ defined in Eq.~(\ref{q_th}), given that the latter is locally nondecreasing at $p^{\rm th}$.
\end{proposition}
\begin{proof}
Let $\hat{T}$ denote the frequency of ones occurring in the population. Then,
\begin{eqnarray}
&&\Pr\Bigl[\hat{p}\leq{}p^{\rm th},\hat{q}\geq{}q^{\rm th}_{N,n,\epsilon}\left(p^{\rm th}\right)\Bigr]\leq{}\Pr\left[\hat{q}\geq{}q^{\rm th}_{N,n,\epsilon}\left(\hat{p}\right)\right]\nonumber \\
&&\leq{}\sum_{p}\Pr\left[\hat{T}=p\right]\Pr\left[\hat{q}\geq{}q^{\rm th}_{N,n,\epsilon}\left(\hat{p}\right)\Big\rvert{}\hat{T}=p\right]\leq{}\epsilon.
\end{eqnarray}

In the first inequality we use the fact that $\{\hat{p}\leq{}p^{\rm th},\hat{q}\geq{}q^{\rm th}_{N,n,\epsilon}\left(p^{\rm th}\right)\}\implies{}\{\hat{q}\geq{}q^{\rm th}_{N,n,\epsilon}\left(\hat{p}\right)\}$. This is so because $q^{\rm th}_{N,n,\epsilon}(x)$ is nondecreasing for $x\in[0,p^{\rm th}]$ if it is nondecreasing at $x=p^{\rm th}$, and the latter is guaranteed by assumption. The averaging over $\hat{T}$ in the second inequality is a purely \emph{ad hoc} step to enable the use of Proposition 2, which presumes a fixed ratio of ones in the population. The third inequality follows from Proposition 2.
\end{proof}

A relevant observation is that the condition on the slope at $p^{\rm th}$, which is equivalent to $\Gamma^{+'}_{n,\epsilon}(p^{\rm th})\geq{}n/N$ [$\Gamma^{+'}_{n,\epsilon}(p^{\rm th})$ denoting the derivative of $\Gamma^+_{n,\epsilon}(x)$ with respect to $x$ evaluated in $p^{\rm th}$], is unrestrictive for all practical purposes~\cite{comment_2}.

On another note, the reader is referred to Appendix II for the failure probability estimation of a decoy-state QKD protocol, which relies not only on Proposition 2 but also on Proposition 1.\\

\textit{Numerical results}---Certainly, tighter results can be obtained by resorting to numerical methods. For instance, using Stirling's formula and the Gaussian cumulative density function, the authors of~\cite{Hayashi_1} construct a HG tail inequality that outperforms our relaxed Chernoff bound in certain regimes, but does not enable a closed formula for the related CI. Alternatively, one could use the exact Chernoff bound rather than a relaxation. Note, however, that if one is willing to use numerical techniques, tail inequalities may become superfluous, because numerically computing the CMF of the HG distribution is often feasible in QKD applications. Crucially, this enables the calculation of  optimal one-sided CIs~\cite{Weizhen,Buonaccorsi} using the well-known Clopper-Pearson method based on pivoting the CMF~\cite{Clopper}. In particular, for our purposes, one can replace $\Gamma^{+}_{n,\epsilon}(\hat{p})$ in Corollary 1 by the tighter statistic $\mathcal{CP}^{+}_{N,n,\epsilon}(\hat{p})$, defined via~\cite{Wright,Weizhen,Buonaccorsi}
\begin{equation}\label{CP}
\mathcal{CP}^{+}_{N,n,\epsilon}\left(x\right)=\min\left\{p\geq{}x\big\rvert{}\Pr\bigl[\hat{p}\leq{}x|p\bigr]\leq{}\epsilon\right\}
\end{equation}
if the target set of the minimization is nonempty, and via $\mathcal{CP}^{+}_{N,n,\epsilon}(x)=(N+1)/N$ otherwise~\cite{clarification}, where we recall that $p$ takes discrete values in the set $\{0,1/N,2/N,\ldots{},1\}$~\cite{speed-up}. Remarkably, accurate determination of the CMF is required in order to compute $\mathcal{CP}^{+}_{N,n,\epsilon}(\hat{p})$, with a numerical precision beyond $\epsilon$ according to Eq.~(\ref{CP})~\cite{precise}. Indeed, for a broad range of settings of interest to QKD, this is a manageable task with existing software.

Importantly, since $\mathcal{CP}^{+}_{N,n,\epsilon}(\hat{p})$ fulfills Corollary 1, Proposition 2 also holds replacing Eq.~(\ref{q_th}) by
\begin{equation}\label{q_th_2}
q^{\rm th}_{N,n,\epsilon}(x)=\frac{N\mathcal{CP}^{+}_{N,n,\epsilon}(x)-nx}{N-n},
\end{equation}
and monotonicity arguments allow to further establish Proposition 3 with this alternative threshold function.

In Fig.~\ref{fig:qvsp} we compare various candidate threshold values $q^{\rm th}$ satisfying Eq.~(\ref{goal}): Eq.~(\ref{q_th}), Eq.~(\ref{q_th_2}), and the thresholds derived in~\cite{Ekert} and~\cite{Leverrier}. For illustration purposes, we consider typical QKD inputs, $p^{\rm th}\leq{}4\%$, $N=10^5$, $n=10^{4}$ and $\epsilon=10^{-9}$, and remark that the improvement of our tools---Eq.~(\ref{q_th}) and Eq.~(\ref{q_th_2})---prevails in all the explored parameter regimes relevant for QKD.

\begin{figure}[H]
    \centering
    \includegraphics[width=\columnwidth]{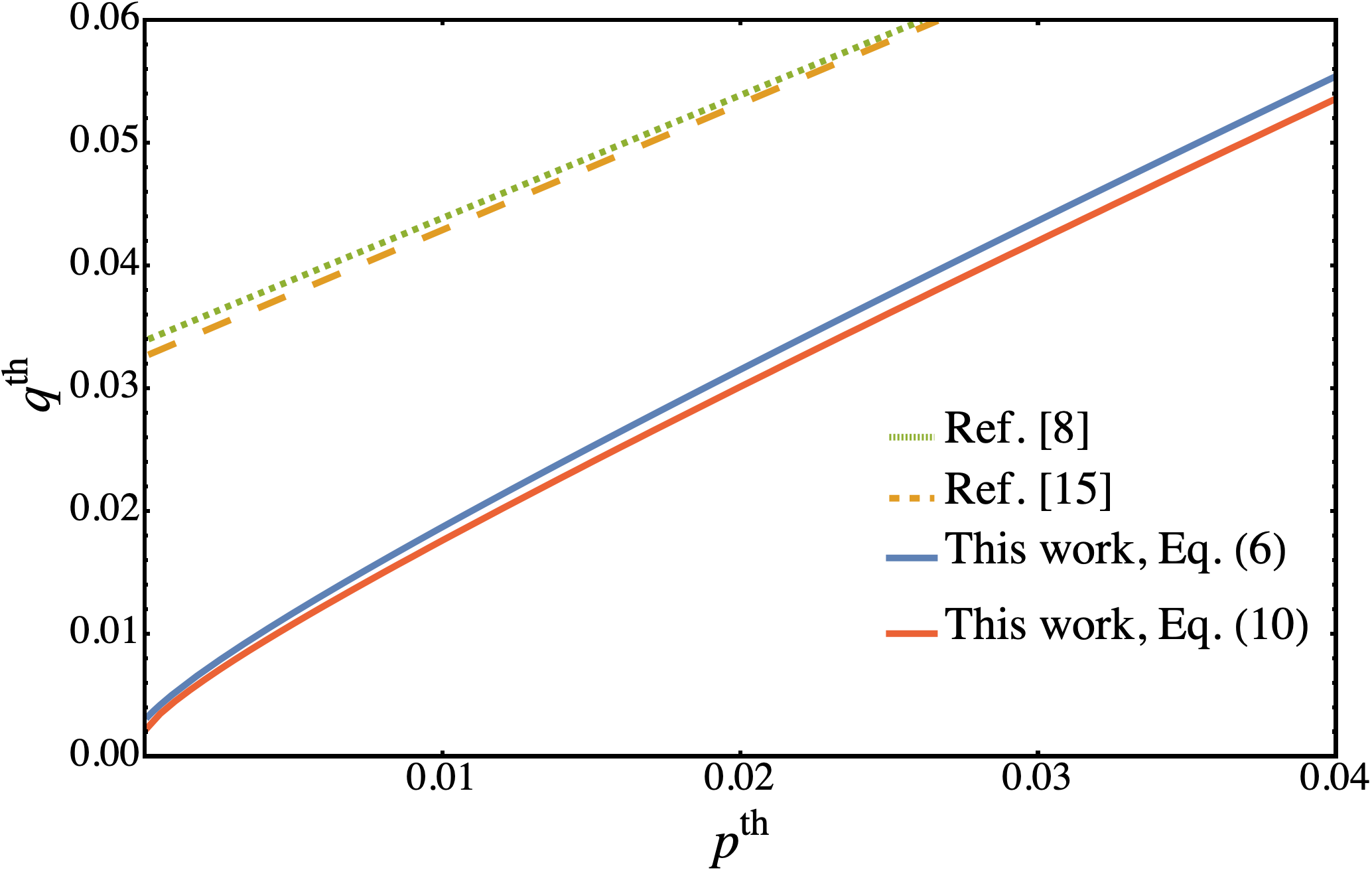}
    \caption{Threshold value $q^{\rm th}$ satisfying Eq.~(\ref{goal}), for $p^{\rm th}\leq{}4\%$, $N=10^5$, $n=10^{4}$,  and $\epsilon=10^{-9}$. The solid red (blue) line corresponds to our numerical (analytical) tool, based on the exact cumulative mass function of the hypergeometric distribution (our relaxed additive Chernoff bound~\cite{Chernoff}). The dashed orange line uses the result in~\cite{Ekert}, which combines Serfling~\cite{Serfling} and Hush and Scovel~\cite{Hush} inequalities, and the dotted green line uses the result in~\cite{Leverrier}, based on Serfling inequality alone.}
    \label{fig:qvsp}
\end{figure}

Additionally, in Appendix IV we include optimal one-sided CIs for the average of independent Bernoulli parameters~\cite{BancalSekatski,Mattner}, potentially replacing $\Gamma^{\pm}_{n,\epsilon}(\hat{p})$ (Proposition 1) in the failure probability estimation of decoy-state QKD schemes. These CIs are related to the Clopper-Pearson CIs for binomial proportions, which in fact have been used in the context of QKD security before (see e.g.~\cite{Lucamarini_1,Lucamarini_2}).\\

\textit{QKD simulations}---Next, we show the impact that the proposed tools have in the finite-key performance of QKD. Let us consider the ideal BBM92 protocol first. Consistently with~\cite{Ekert,Leverrier}, we evaluate a simplified scenario in which Alice and Bob share a $2N$-partite quantum state as a protocol input. Then, they jointly select between two binary and complementary quantum measurements uniformly at random on a round-by-round basis, each of them extracting a raw data block of $N$ bits. These blocks are randomly split into a fixed-length test (sifted-key) string of $n$ ($N-n$) bits. That is to say, both the key and the test string are randomly sampled from the two bases. In the simulations, we set the tolerated quantum bit error rate (QBER) to $p^{\rm th}=4.55\%$ for ease of comparison with~\cite{Ekert}, where this choice is motivated by the results of the Micius experiment~\cite{Micius}. Also, an error correction leakage of $\lambda_{\rm EC}=1.19(N-n)h(p^{\rm th})$ bits is assumed, $h(x)$ denoting the binary entropy of $x$. The correctness parameter, $\epsilon_{\rm cor}$, and the privacy amplification error, $\epsilon_{\rm PA}$, are both set to $10^{-8}$ for illustration purposes, and the failure probability bound is set to $\epsilon_{\rm PE}=4\times{}10^{-16}$. Altogether, these settings lead to an overall secrecy parameter of $\epsilon_{\rm sec}=5\times{}10^{-8}$  (see Appendix I).

\begin{figure}[htbp!]
\centering
 \begin{subfigure}{.5\textwidth}
   \centering
   \includegraphics[width=\columnwidth]{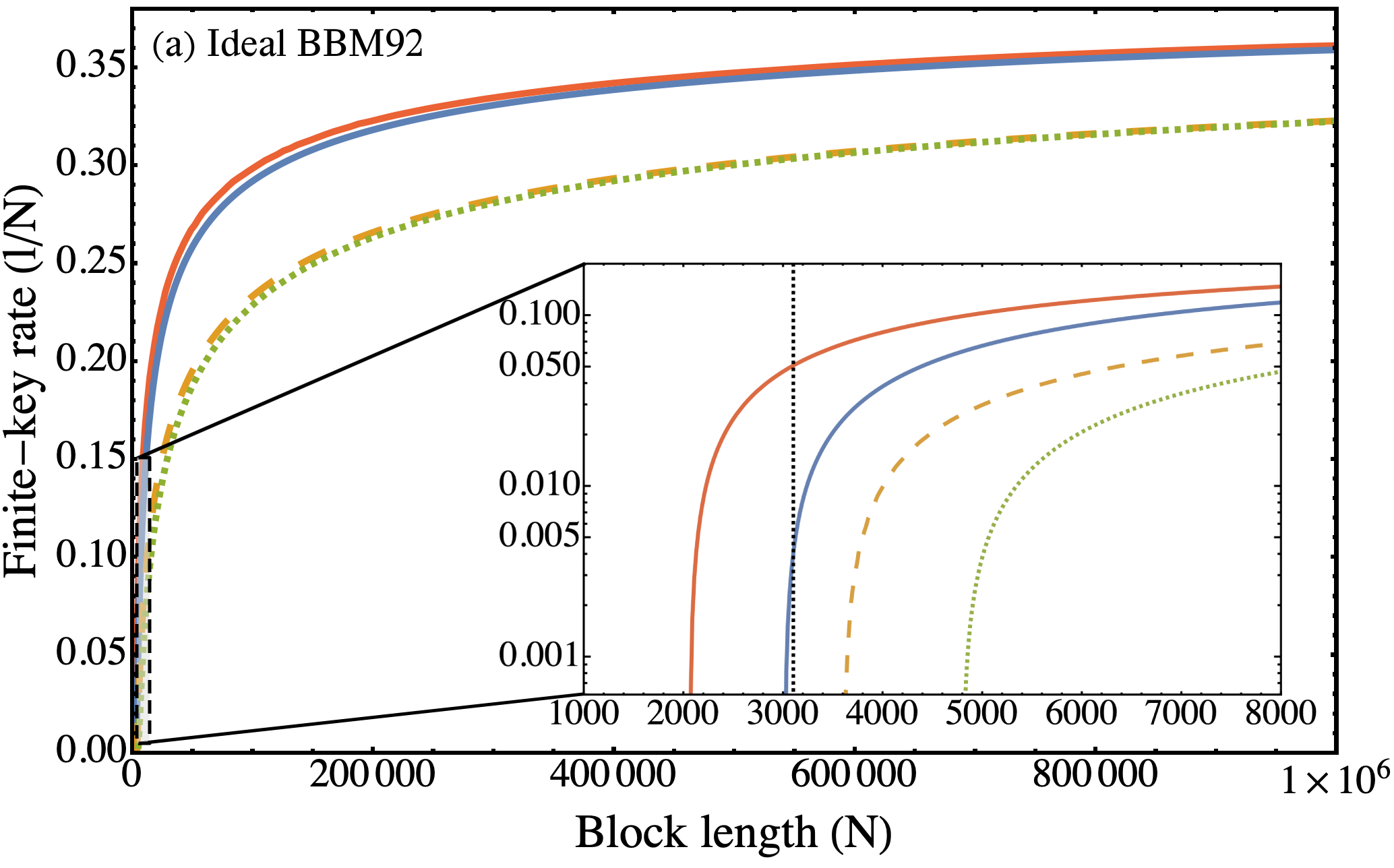}
   \end{subfigure}%
  \hfill
 \begin{subfigure}{.5\textwidth}
   \centering
   \includegraphics[width=\columnwidth]{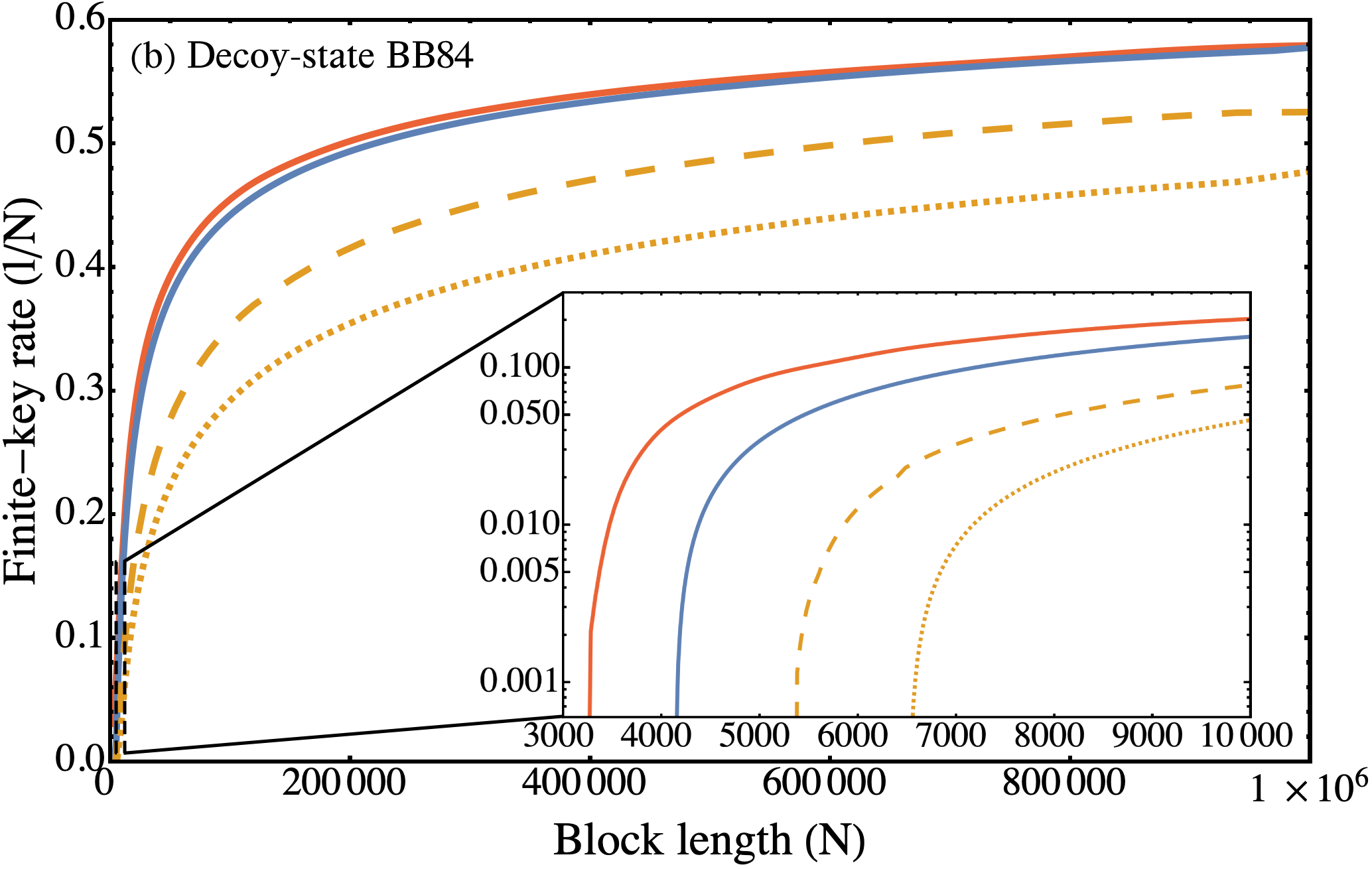}
    \end{subfigure}
    \caption{Finite secret key rate versus data block size. (a) Ideal BBM92 protocol. The dotted green line uses the result in~\cite{Leverrier}, the dashed orange line uses the result in~\cite{Ekert} and the solid blue line uses Proposition 3 in this work, originating from our relaxed additive Chernoff bound. Last, the solid red line corresponds to the numerical upgrade of Proposition 3, which replaces the threshold function of Eq.~(\ref{q_th})---relaxed-Chernoff---by that of Eq.~(\ref{q_th_2})---Clopper-Pearson---. The dotted black vertical line in the inset indicates the block size gathered in the Micius experiment~\cite{Micius}, $N=3100$ bits. (b) Decoy-state BB84 protocol. Both orange lines use the random sampling tool of~\cite{Ekert}, but the dotted (dashed) line relies on Hoeffding's inequality~\cite{Hoeffding} (the tighter multiplicative Chernoff bound deployed in~\cite{Zhang}) for the extra confidence intervals for Bernoulli trials required by the decoy-state method. The solid blue line uses the random sampling tool of Proposition 2 and the extra confidence intervals of Proposition 1, both based on our relaxed Chernoff bound. Finally, the solid red line uses our numerical upgrades of these two propositions, deploying optimal confidence intervals in both cases~\cite{Buonaccorsi,BancalSekatski}. The protocol inputs and experimental settings considered in the simulations are listed in the main text.}
    \label{decoyplot}
\end{figure}

With these inputs, in Fig.~\ref{decoyplot}(a) we plot the extractable key length divided by the data block size for blocks up to $N=10^{6}$ bits, optimizing the sampling fraction $n/N$. In addition, we superpose an inset of the minimum block-size regime, showing that our analytical tool (solid blue line) reduces the minimum block size with respect to~\cite{Ekert} (dashed orange line) by more than 16\%, while our numerical tool (solid red line) attains a reduction beyond 40\%. Furthermore, the tool of~\cite{Ekert} would roughly require a 3-orders-of-magnitude larger $\epsilon_{\rm PE}$ for a nonzero key with the block size collected in the Micius experiment~\cite{Micius} (dotted black vertical line in the inset). Notably as well, both of our tools perform very similarly---and significantly above~\cite{Ekert,Leverrier}---for noncritical block sizes.

Let us now assess the more practical decoy-state BB84 protocol. We consider a scheme with three common intensities per basis, $\mu$ (signal), $\nu$ (decoy) and $\omega$ (vacuum), and employ the $Z$ ($X$) basis for key extraction (PE) with a fixed-size data block of $N_{\rm Z}$ ($N_{\rm X}$) bits drawn at random from the corresponding data pool. Again, in the simulations we contemplate blocks up to $N=N_{\rm Z}+N_{\rm X}=10^{6}$ bits, and superpose an inset of the minimum block-size regime. Naturally, we set $N_{\rm Z}=q_{\rm Z}^2/(q_{\rm Z}^2+q_{\rm X}^2)N$, where $q_{\rm Z}$ ($q_{\rm X}$) denotes the $Z$ ($X$) basis probability, common to both parties. Also, we fix $\omega=10^{-4}$ to account for the finite extinction ratio of the intensity modulator~\cite{finite_mdiQKD}. The settings $\mu$ and $\nu$, together with their respective probabilities, $p_{\mu}$ and $p_{\nu}$, and the test basis probability, $q_{\rm X}$, are tuned to optimize the key length per block, $l/N$. For simplicity, in order to select the acceptance thresholds of the PE tests~\cite{novelty}, the average behavior of a typical fiber and detector model is considered (see Appendix II for the details). We assume an overall system loss of 30 dB, a dark count probability of $p_{\rm d}=6\times{}10^{-7}$~\cite{nino12}, and a misalignment error rate of $e_{\rm mis}=5 \times10^{-3}$~\cite{Lim}. As in the BBM92 protocol, the error correction leakage is set to $\lambda_{\rm EC}=1.19N_{\rm Z}h(\theta^{\rm th})$, $\theta^{\rm th}$ denoting the correctable QBER of a key block~\cite{clarification2}. Last, we set again the overall failure probability to $\epsilon_{\rm PE}=4 \times 10^{-16}$, and assume a common value for each individual error term contributing to it~\cite{clarification3}. Similarly, we take $\epsilon_{\rm cor}=\epsilon_{\rm PA}=\delta=10^{-8}$, where $\delta>0$ is an intrinsic variable of the security analysis (see Appendix II). Overall, this leads to $\epsilon_{\rm sec}={}7\times{}10^{-8}$.

The results are shown in Fig.~\ref{decoyplot}(b). Unlike the BBM92 protocol, the security of the decoy-state scheme does not rely on the random sampling tool alone, but also on CIs for the average parameter of a set of Bernoulli trials. In fact, such CIs are repeatedly used in the decoy-state analysis, and thus they have a critical impact in the performance. The reader is referred to the caption for the specifications of the different lines in Fig.~\ref{decoyplot}(b). Overall, our analytical tools reduce the minimum block size by more than $22\%$ compared to previous approaches, and the reduction exceeds $38\%$ with the numerical tools. Once more, our analytical and numerical solutions perform very similarly in noncritical regimes, with a strong edge over existing approaches.

On top of all the above results, in Appendices V and VI we survey analytical HG bounds available in the literature, and derive two tools that perform better than our Chernoff-based confidence upper bound (Corollary 1) for small blocks. Particularly, in this regime, we obtain a performance comparable to the numerical one by reproducing our analytical method but with Hush and Scovel inequality.\\

\textit{Conclusions}---The performance of QKD is deeply affected by finite statistics. In satellite QKD~\cite{Micius,Sidhu,Shields,Bedington}, the data collection per satellite overpass is limited, and sharp statistics may entail an enormous saving of time and resources. Moreover, one could envision future applications requiring short keys ondemand for daily transactions with high security standards. For these, a poor statistical analysis would result in an unbearable penalty.

In this work, we have developed statistical tools of unprecedented tightness among QKD security proofs~\cite{EUR}. Although we devise the tools for the BBM92 and the decoy-state BB84 protocols, they are of use in, e.g., measurement-device-independent QKD~\cite{finite_mdiQKD,liu2023experimental} and twin-field QKD~\cite{curras2021tight,liu2023TF} as well. Technically, the key enablers of our main analytical results are a classical link between random sampling with and without replacement~\cite{Hoeffding}, and a careful relaxation of the long-known additive Chernoff bound~\cite{Chernoff}. Arguably, the tools might be of interest beyond the realm of QKD, in applications such as quality control tests, clinical trials or hypothesis testing in general. To support this claim, in Appendices V and VI we thoroughly compare our relaxed Chernoff bound with existing hypergeometric bounds suited to the same purpose. What is more, for Hush and Scovel inequality and a tight Bernstein-type~\cite{Bernstein} inequality due to Greene and Wellner~\cite{Greene}, we further reproduce our analytical method and obtain new analytical tools that outperform our Chernoff-based confidence bound if small enough block sizes are considered.

On top of it, we have put forward the possibility of replacing tail inequalities---the cornerstone of statistical analyses in QKD today---by exact computational methods in statistical inference. Remarkably, our results reveal the game-changing potential of such methods to achieve QKD with minimum block sizes, and reinforce the tightness of our simpler analytical tools when moderate-to-large blocks are considered.

\textit{Acknowledgements}---We acknowledge discussions with Jay Bartroff, Davide Rusca and Jean-Daniel Bancal, and useful feedback from Jon Wellner. This work was supported by the Galician Regional Government (consolidation of Research Units: AtlantTIC), the Spanish Ministry of Economy and Competitiveness (MINECO), the Fondo Europeo de Desarrollo Regional (FEDER) through the Grant No. PID2020-118178RB-C21, MICIN with funding from the European Union NextGenerationEU (PRTR-C17.I1) and the Galician Regional Government with own funding through the “Planes Complementarios de I+D+I con las Comunidades Autónomas” in Quantum Communication, the European Union’s Horizon Europe Framework Programme under the Marie Sklodowska-Curie Grant No. 101072637 (Project QSI) and the project Quantum Secure Networks Partnership
(QSNP) (Grant No. 101114043). M.C. acknowledges support from a ``Salvador de Madariaga'' grant from the Spanish Ministry of Science, Innovation and Universities (Grant No. PRX22/00192). 
\newpage

\onecolumngrid
\section*{Appendix I: ideal BBM92 protocol}\label{I}
\subsection{Protocol description}
\begin{table}[b]
\caption{BBM92 protocol inputs}
\label{bbm92inputs}
\begin{tabular}{|c|l|}
\hline
$N$ & Raw data block size  \\
\hline
$n$ & Test data block size \\
\hline
$p^{\rm th}$ & Threshold value for the  QBER\\
\hline
$\lambda_{\rm EC}$ & Error correction leakage \\
\hline
$\epsilon_{\rm PE}$ & Failure probability of the parameter estimation step\\
\hline
$\epsilon_{\rm cor}$ & Failure probability of the error correction step\\
\hline
$\epsilon_{\rm PA}$ & Failure probability of the privacy amplification step\\
\hline
\end{tabular}
\end{table}
In this section we describe the ideal BBM92 protocol from Ref.~\cite{Ekert}, where no detection losses are considered and both the key string and the test string are randomly sampled from the two bases. In this scheme, Alice and Bob agree on the protocol inputs specified in Table \ref{bbm92inputs}, \emph{i.e.} the parameters $N,n,p^{\rm th}$, $\lambda_{\rm EC},\epsilon_{\rm PE},\epsilon_{\rm cor}$, and $\epsilon_{\rm PA}$. Moreover, they share a 2$N$-partite quantum state and consistently choose between two complementary measurements on a round-by-round basis, sacrificing a random subset of $n$ measurement outcomes for testing.

The protocol runs as follows:
\begin{enumerate}
    \item \emph{Measurement:} Alice and Bob measure their quantum states with two complementary measurements, which they select jointly in each round, obtaining a  raw data block of size $N$. A sifted key string of length $N-n$ is sampled at random, and the rest of the raw data (\emph{i.e.} $n$ bits) constitutes the test string.
    \item \emph{Parameter estimation} (PE):
They publicly reveal the test data and compute the quantum bit error rate (QBER). If the QBER exceeds a tolerated threshold value $p^{\rm th}$, they abort the protocol.

  \item{}\emph{Error correction} (EC): {They run a pre-defined EC protocol to reconcile their sifted keys, revealing $\lambda_{\rm EC}$ secret-key bits at most.}

\item{}\emph{Error verification} (EV): {They perform an EV step based on 2-universal hashing, using  tags of $\log(2/\epsilon_{\rm cor})$ bits at most. If the EV tags do not match, they abort the protocol. Otherwise, they proceed to privacy amplification.}

\item{}\emph{Privacy amplification} (PA): {The parties perform a PA step based on 2-universal hashing, each of them obtaining a final key of $l$ bits, with
\begin{equation}
\label{bbm92key}
    l=\left\lfloor (N-n)\left\{ 1-h\left[q^{\rm th}_{N,n,\epsilon_{\rm PE}}\left(p^{\rm th}\right)\right]\right\}-\lambda_{\rm EC}-\log\left(\frac{1}{2 \epsilon_{\rm cor} \epsilon_{\rm PA}^2}\right)\right\rfloor.
\end{equation}
Here, $h(x)$ is the binary entropy function, and the threshold function $q^{\rm th}_{N,n,\epsilon_{\rm PE}}(x)$  depends on the random sampling tool used for PE. These are included in the next section. On the other hand, $\lambda_{\rm EC}$ is the EC leakage, $\epsilon_{\rm cor}$ is the correctness parameter and $\epsilon_{\rm PA}$ is the PA error.}
\end{enumerate}

Importantly, this approach yields provably $\left(2\sqrt{\epsilon_{\rm PE}}+\epsilon_{\rm PA}\right)$-secret and $\epsilon_{\rm cor}$-correct output keys in virtue of a PA lemma for the smooth min-entropy~\cite{QLHL}, where $\epsilon_{\rm PE}$ denotes the corresponding smoothing parameter. As shown in~\cite{Leverrier}, $\epsilon_{\rm PE}$ directly provides an upper bound on the failure probability of the protocol, precisely addressed in the next section.

\subsection{Failure probability}
We recall that, in the BBM92 protocol above, a PE failure refers to the joint event where the PE test succeeds and the presumed threshold on the phase-error rate (PHER) is incorrect~\cite{Leverrier,Ekert}. The threshold $q^{\rm th}_{N,n,\epsilon_{\rm PE}}\left(p^{\rm th}\right)$ on the PHER can be tuned to guarantee that the failure probability of the protocol is below $\epsilon_{\rm PE}$. Namely, it can be chosen to enforce $\Pr\left[\hat{p}\leq{}p^{\rm th},\hat{q}\geq{}q^{\rm th}_{N,n,\epsilon_{\rm PE}}\left(p^{\rm th}\right)\right]\leq{\epsilon_{\rm PE}}$, where $\hat{p}$ and $\hat{q}$ respectively denote the QBER and the PHER. Below we provide suitable candidates for the threshold function based on different statistical tools.
\subsubsection{Relaxed additive Chernoff bound}
From Proposition 3 in the main text, we have
\begin{equation}\label{threshold}
q^{\rm th}_{N,n,\epsilon_{\rm PE}}\left(p^{\rm th}\right)=\frac{N\Gamma^+_{n,\epsilon_{\rm PE}}\left(p^{\rm th}\right)-n{}p^{\rm th}}{N-n},
\end{equation}
where $\Gamma^+_{n,\epsilon_{\rm PE}}(x)$ is given in Proposition 1 of the main text. This corresponds to the solid blue line of Fig.~\ref{decoyplot}(a) there.
\subsubsection{Serfling and Hush \& Scovel inequalities combined}
From Lemma 1 in~\cite{Ekert}, we have
\begin{equation}\label{LimEkert}
  q^{\rm th}_{N,n,\epsilon_{\rm PE}}\left(p^{\rm th}\right)=\max_{\xi}\left\{p^{\rm th}+\xi+\frac{1}{N-n}\sqrt{1+\frac{1}{2 g}\ln\left[\frac{1}{\epsilon_{\rm PE}-\displaystyle{\exp \left(-\frac{2 N n \xi^2}{N-n+1}\right)}}\right]}\right\}
\end{equation}
with
\begin{equation}
    g=\frac{1}{N\left(p^{\rm th}+\xi\right)+1}+\frac{1}{N\left(1-p^{\rm th}-\xi\right)+1},
\end{equation}
where the maximization runs over all $\xi>0$ such that $N(p^{\rm th}+\xi)\in\mathbb{Z}^{+}$ and $\epsilon_{\rm PE}>\displaystyle{}{\exp \left(-\frac{2 N n \xi^2}{N-n+1}\right)}$.

As mentioned in the main text, this approach provides a combination of Serfling inequality~\cite{Serfling} with Hush \& Scovel inequality~\cite{Hush} based on the union bound, and it corresponds to the dashed orange line of Fig.~\ref{decoyplot}(a) in the main text.

\subsubsection{Serfling inequality}
Using Lemma 6 from~\cite{Leverrier} (which follows from Serfling inequality~\cite{Serfling}), we have
\begin{equation}\label{Serfling_threshold}
q^{\rm th}_{N,n,\epsilon_{\rm PE}}\left(p^{\rm th}\right)=p^{\rm th}+\sqrt{\frac{N(n+1)\ln(\epsilon_{\rm PE}^{-1})}{2(N-n)n^2}}.
\end{equation}
This corresponds to the dotted green line of Fig.~\ref{decoyplot}(a) in the main text.

\subsubsection{Exact hypergeometric cumulative mass function}
In the main text, a numerical confidence upper bound $\mathcal{CP}^{+}_{N,n,\epsilon}(\hat{p})$ on the population parameter is presented~\cite{Buonaccorsi,Wright}, based on the classical Clopper-Pearson method of pivoting the cumulative mass function~\cite{Clopper}. Since one can replace $\Gamma^+_{n,\epsilon}(\hat{p})$ by $\mathcal{CP}^{+}_{N,n,\epsilon}(\hat{p})$ in Corollary 1 of the main text, $\mathcal{CP}^{+}_{N,n,\epsilon}(\hat{p})$ automatically fulfills Proposition 2 there as well. From this, one can easily show that it also satisfies Proposition 3, thereby providing another suitable threshold function:
\begin{equation}\label{threshold_CP}
q^{\rm th}_{N,n,\epsilon_{\rm PE}}\left(p^{\rm th}\right)=\frac{N\mathcal{CP}^{+}_{N,n,\epsilon_{\rm PE}}\left(p^{\rm th}\right)-n{}p^{\rm th}}{N-n}.
\end{equation}
This corresponds to the solid red line in of Fig.~\ref{decoyplot}(a) in the main text.

\section*{Appendix II: decoy-state BB84 protocol}\label{II}

We consider a typical decoy-state BB84 protocol in the lines of~\cite{Lim}. Particularly, three common intensities are used in each basis, and the $Z$ ($X$) basis is employed for key extraction (PE).

\subsection{Protocol description}\label{protocol}

Alice and Bob agree on the protocol inputs $T$, $q_{\rm X}$, $\mathcal{K}:=\left\{\mu, \nu, \omega\right\}$, $p_{\mu}$ $p_{\nu}$, $N_{\rm X}$, $N_{\rm Z}$, $\theta^{\rm th}$, $\lambda_{\rm EC}$, $\epsilon_{\rm PE}$, $\epsilon_{\rm cor}$, $\epsilon_{\rm PA}$, and $\delta$, specified in Table~\ref{bb84inputs}. In addition, unless the random sampling tool of~\cite{Ekert} is considered, they further agree on two more protocol inputs inherent to the corresponding PE test, $n_{1,\rm Z}^{\rm th}$ and $\phi_{1,\rm Z}^{\rm th}$. Alternatively, if the tool of~\cite{Ekert} is considered, the additional protocol inputs are $n_{1,\rm Z}^{\rm th,U(L)}$, $n_{1,\rm X}^{\rm th,U(L)}$ and $m_{1,\rm X}^{\rm th,U}$, with which they can ultimately define the relevant threshold values entering the secret key length formula:
\begin{equation}\label{functions}
n_{1,\rm Z}^{\rm th}=n_{1,\rm Z}^{\rm th,L},\hspace{.2cm}\phi_{1,\rm Z}^{\rm th}=\max_{u,v,\in\mathcal{B}}q_{u,v,\epsilon_{\rm PE}}^{\rm th}(e_{1,\rm X}^{\rm th}),
\end{equation}
where $e_{1,\rm X}^{\rm th}=m_{1,\rm X}^{\rm th,U}/n_{1,\rm X}^{\rm th,L}$, $\mathcal{B}=\bigl\{u\in\bigl[n_{1,\rm Z}^{\rm th,L}+n_{1,\rm X}^{\rm th,L},n_{1,\rm Z}^{\rm th,U}+n_{1,\rm X}^{\rm th,U}\bigr],v\in\bigl[n_{1,\rm X}^{\rm th,L},n_{1,\rm X}^{\rm th,U}\bigr],u>v\bigr\},$ and the threshold function under consideration is that of Eq.~(\ref{LimEkert}). Importantly, we adopt a finer-grained PE test for the random sampling tool of~\cite{Ekert} in order to simplify the task of bounding the failure probability in Sec.~\ref{alternative} below. Indeed, it is unclear whether a bound can be derived with this tool without modifying the PE test at all. After all, the tool in~\cite{Ekert} is devised for the BBM92 protocol specifically, rather than for a decoy-state BB84 protocol, and in this work we provide a possible adaptation to the latter.
\\

The protocol runs as follows. For $i=1,\ldots,T$, steps 1 and 2 below are repeated.

\begin{enumerate}

\item{}\emph{State preparation:} Alice chooses a bit value $y_i$ uniformly at random, a basis setting $a_{i}\in \{Z,X\}$ with probability $q_{a_{i}}$, and  an intensity {setting $k_i\in{}\left\{\mu, \nu, \omega\right\}=:\mathcal{K}$ with probability $p_{k_{i}}$}. She prepares a {phase-randomized weak coherent pulse (PRWCP) encoded with the above settings and sends it through the quantum channel to Bob.}

\item{}\emph{Measurement:} Bob measures the {incident pulse in the basis $b_i\in \{Z,X\}$ with probability $q_{b_{i}}$. He records the measurement outcome as $y_i^{\prime}\in\{0,1, \emptyset\}$, where $\emptyset$ refers to the ``no-click" event. Multiple clicks are assigned to 0 or 1 uniformly at random.}
\end{enumerate}
The post-processing and the public discussion run as follows:

\begin{enumerate}
\setcounter{enumi}{2}
\item{}\emph{{Sifting:}} {The bases and intensity settings are publicly revealed and Alice and Bob identify the sets $\mathcal{Z}=\bigl\{i: a_i=b_i=Z,\ y_i^{\prime} \neq \emptyset\bigr\}$ and $\mathcal{X}=\left\{i: a_i=b_i=X,\ y_i^{\prime} \neq \emptyset\right\}$. $N_{\rm Z} $ rounds are drawn at random from $\mathcal{Z}$ to form the sifted key data, $\mathcal{Z}^{\rm sift}$, which decomposes as $\mathcal{Z}^{\rm sift}=\cup_{k \in \mathcal{K}} \mathcal{Z}^{\rm sift}_{k}$ with $\mathcal{Z}^{\rm sift}_k=\left\{i\in\mathcal{Z}^{\rm sift}: k_i=k\right\}$, of size $n_{{\rm Z},k}$. $N_{\rm X} $ rounds are drawn at random from $\mathcal{X}$ to form the test data, $\mathcal{X}^{\rm test}$, which decomposes as $\mathcal{X}^{\rm test}=\cup_{k \in \mathcal{K}} \mathcal{X}^{\rm test}_{k}$ with $\mathcal{X}^{\rm test}_k=\left\{i\in\mathcal{X}^{\rm test}: k_i=k\right\}$,  of size $n_{{\rm X},k}$.}

\item{}\emph{Parameter estimation} (PE): {For each $k \in \mathcal{K}$, they disclose the bit values in $\mathcal{X}_k^{\rm test}$ and compute the corresponding numbers of bit errors, $m_{{\rm X},k}$. With the available data, they perform the PE test that matches their preferred random sampling tool (see Sec.~\ref{failure_decoy} for the details). If the test fails, they abort the protocol.}

\item{}\emph{Error correction} (EC): {They run a pre-defined EC protocol to reconcile their sifted keys, revealing $\lambda_{\rm EC}$ secret key bits at most.}

\item{}\emph{Error verification} (EV): {They perform an EV step based on 2-universal hashing, using tags of $\log(2/\epsilon_{\rm cor})$ bits at most.  If the EV tags do not match, they abort the protocol. Otherwise, they proceed to privacy amplification.}

\item{}\emph{Privacy amplification} (PA): {The parties perform a PA step based on 2-universal hashing, in so obtaining final keys of length
\begin{equation}
\label{keylength}
l=\left\lfloor n^{\rm th}_{1,\mathrm{Z}}\left[1-h\left(\phi^{\rm th}_{1,\rm Z}\right)\right]-\lambda_{\rm EC}-\log\left(\frac{1}{2\epsilon_{\rm cor}\epsilon_{\rm PA}^{2}\delta}\right)\right\rfloor.
\end{equation}
}
\end{enumerate}
This approach yields provably $\left[2\left(\sqrt{\epsilon_{\rm PE}}+\delta\right)+\epsilon_{\rm PA}\right]$-secret~\cite{finpassec,QLHL} and $\epsilon_{\rm cor}$-correct keys for the requested failure probability bound, $\epsilon_{\rm PE}$, carefully addressed in Sec.~\ref{failure_decoy}.
\begin{table}[htbp]
\caption{Decoy-state BB84 protocol inputs alien to the PE test}
\label{bb84inputs}
\begin{tabular}{|c|l|}
\hline
$T$ & Number of transmission rounds  \\
\hline
$q_{\rm X}$ & Test basis probability\\
\hline
$\mathcal{K}:=\{\mu,\nu,\omega\}$ & Set of intensities  \\
\hline
$p_{ \mu}$ & Probability of setting $\mu$ \\
\hline
$p_{ \nu}$ & Probability of setting $\nu$ \\
\hline
$N_{\rm X}$ & Test data block size  \\
\hline
$N_{\rm Z}$ & Key data block size  \\
\hline
$\theta^{\rm th}$ & Threshold value for the sifted-key QBER\\
\hline
$\lambda_{\rm EC}$ & Error correction leakage \\
\hline
$\epsilon_{\rm PE}$ & Failure probability of the parameter estimation step\\
\hline
$\epsilon_{\rm cor}$ & Failure probability of the error correction step\\
\hline
$\epsilon_{\rm PA}$ & Failure probability of the privacy amplification error\\
\hline
$\delta$ & Slack variable in the entropic chain rule~\cite{Vitanov}\\
\hline
\end{tabular}
\end{table}

\subsection{Decoy-state analysis}\label{security}
Here, we keep the analytical decoy-state bounds provided in \cite{Lim} and generalize the statistical analysis~\cite{counterfactual}, assuming a common error probability $\varepsilon$ per statistical bound for simplicity. 

The number of single-photon counts in $\mathcal{Z}^{\rm sift}$ is lower-bounded as
\begin{equation}
\label{singlow}
n_{1,\mathrm{Z}} \overset{  3 \varepsilon}{\geqslant} n^{\rm L}_{1,\mathrm{Z}}:= \frac{\tau_1 \mu}{\mu\left(\nu-\omega\right)-\nu^2+\omega^2}\left(\frac{e^{\nu}}{p_{\nu}}n_{\mathrm{Z}, \nu}^{-}-\frac{e^{\omega}}{p_{\omega}}n_{\mathrm{Z}, \omega}^{+}-\frac{\nu^2-\omega^2}{\mu^2} \frac{e^{\mu}}{p_{\mu}}n_{\mathrm{Z}, \mu}^{+}\right),
\end{equation}
where $\tau_1=\sum_{k \in \mathcal{K}} e^{-k} k \,p_{k}$ and $n_{\mathrm{Z}, k}^{ \pm}=B^{\pm}\left( \varepsilon,n_{{\rm Z},k},N_{\rm Z}\right)$ for the functions $B^{\pm}\left( \varepsilon,n_{{\rm Z},k},N_{\rm Z}\right)$ of Sec.~\ref{addbounds} to Sec.~\ref{CPbounds}, which depend on the tail inequality under consideration for Bernoulli sampling. Importantly as well, the superscript over the inequality symbol in Eq.~(\ref{singlow}) indicates that the corresponding bound holds except with probability $3\varepsilon$ at most. This follows from Boole's inequality~\cite{bonferroni1936teoria}, noticing that all three statistical estimates have error probability $\varepsilon$ at most. We recall that this notation is extensively used below for convenience.
 
 Complementarily to Eq.~(\ref{singlow}), it follows that
\begin{equation}
\label{singup}
n_{1,\rm Z}\overset{2 \varepsilon}{\leqslant} n^{\rm U}_{1,\rm Z}:=\frac{\tau_1}{{\nu-\omega}} \left(\frac{e^{\nu}}{p_{\nu}} n_{\rm Z, \nu}^{+}-\frac{e^{\omega}}{p_{\omega}} n_{\rm Z, \omega}^{-}\right).
\end{equation}

Replacing $Z$ by $X$ everywhere in Eqs.~(\ref{singlow}) and (\ref{singup}), we reach the corresponding bounds for the single-photon counts in $\mathcal{X}$, $n_{1,\mathrm{X}} \overset{  3 \varepsilon}{\geqslant} n^{\rm L}_{1,\mathrm{X}}$ and $n_{1,\mathrm{X}} \overset{  2 \varepsilon}{\leqslant} n^{\rm U}_{1,\mathrm{X}}$.

The number of bit errors associated with the single-photon events in $\mathcal{X}$ is upper bounded as
\begin{equation}
\label{phase}
m_{1,\rm X}\overset{2 \varepsilon}{\leqslant} m^{\rm U}_{1,\rm X}:= \frac{\tau_1}{{\nu-\omega}} \left(\frac{e^{\nu}}{p_{\nu}} m_{\rm X, \nu}^{+}-\frac{e^{\omega}}{p_{\omega}} m_{\rm X, \omega}^{-}\right),
\end{equation}
 where $m^{\pm}_{{\rm X},k}=B^{\pm}\left( \varepsilon,m_{{\rm X},k},M_{\rm X}\right)$ and $M_{\rm X}=\sum_{k\in \mathcal{K}}m_{{\rm X},k}$.
 
 Lastly, it follows that
\begin{equation}
\label{errlow}
m_{1,\mathrm{X}} \overset{  3 \varepsilon}{\geqslant} m^{\rm L}_{1,\mathrm{X}}:= \frac{\tau_1 \mu}{\mu\left(\nu-\omega\right)-\nu^2+\omega^2}\left(\frac{e^{\nu}}{p_{\nu}}m_{\mathrm{X}, \nu}^{-}-\frac{e^{\omega}}{p_{\omega}}m_{\mathrm{X}, \omega}^{+}-\frac{\nu^2-\omega^2}{\mu^2} \frac{e^{\mu}}{p_{\mu}}m_{\mathrm{X}, \mu}^{+}\right).
\end{equation}

Below we provide the $B^{\pm}(x,y,z)$ functions that we contemplate.

\subsubsection{Relaxed additive Chernoff bound}
\label{addbounds}
First, we include the confidence bounds that arise from our tight relaxation of the additive Chernoff bound~\cite{Chernoff}. Precisely, it follows from Proposition 1 in the main text that one can take
\begin{equation}\label{add_pb}
B^{\pm}(\varepsilon,n\hat{p},n ):=n \Gamma^{\pm}_{n, \varepsilon}(\hat{p})
\end{equation}
for the statistics $\Gamma^{\pm}_{n, \varepsilon}(\hat{p})$ defined therein.
 \subsubsection{Exact multiplicative Chernoff bound}
 \label{multbounds}
Alternatively, one could follow the approach in~\cite{Zhang}, which is based on the multiplicative form of the Chernoff bound. As made explicit in~\cite{Bahrani2019,ZC21}, we have
\begin{equation}\label{mult_pb}
    B^{\pm}(\varepsilon,n\hat{p},n ):=n\hat{p}\pm\delta^{\pm}(n\hat{p},\varepsilon),
\end{equation}
for
 \begin{equation}
  \delta^+(x,y)=x\left\{ W_0\left[-\exp(-c_{x,y})\right] + 1 \right\},
   \end{equation}
    \begin{equation}
  \delta^-(x,y)=\begin{cases}
-x \left\{ W_{-1}\left[-\exp(-c_{x,y})\right] + 1 \right\} & \text{if } x \neq 0, \\
\ln{y^{-1}} & \text{if } x = 0
\end{cases},
 \end{equation}
 where $W_j$ stands for the $j$-th branch of the  Lambert $W$ function and $c_{x,y}$ is defined as $c_{x,y} = 1 + \ln{(1/y)}/x$.

\subsubsection{Hoeffding inequality}\label{hoeffbounds}
Similarly, from Hoeffding inequality~\cite{Hoeffding}, we have that
\begin{equation}\label{hoeffding_pb}
    B^{\pm}(\varepsilon,n\hat{p},n ):=n\hat{p}\pm\sqrt{\frac{n}{2}\ln\frac{1}{\varepsilon}}.
\end{equation}

\subsubsection{Exact binomial cumulative mass function}
\label{CPbounds}
Lastly, we include the optimal one-sided confidence intervals presented in Sec.~\ref{IV}, due to~\cite{BancalSekatski}. These rely on the exact evaluation of the binomial cumulative mass function, which is determined by the regularized beta function. Precisely, one can take
\begin{equation}\label{CP_pb}
    B^{\pm}(\varepsilon,n\hat{p},n ):=n\left[\mathcal{F}^{\pm}_{n,\varepsilon}(\hat{p})\pm\Delta\right]
\end{equation}
for any $\Delta>0$, where the statistics $\mathcal{F}^{\pm}_{n,\varepsilon}(\hat{p})$ are introduced in Theorem 1 of Sec.~\ref{IV}.

\subsection{Failure probability}\label{failure_decoy}
In this section, we provide the failure probability estimation of the decoy-state BB84 protocol for the different random sampling tools under consideration.

\subsubsection{Relaxed additive Chernoff bound}\label{ours}
In this case, the PE test is defined by the success event
\begin{equation}\label{simpler}
\Omega_{\rm TEST} = \{n_{1,\rm Z}^{\rm L} \geq n_{1,\rm Z}^{\rm th},\ \phi_{1,\rm Z}^{\rm U} \leq \phi_{1,\rm Z}^{\rm th}\},
\end{equation}
where the threshold values $n_{1,\rm Z}^{\rm th}$ and $\phi_{1,\rm Z}^{\rm th}$ are the test-specific protocol inputs (see Sec.~\ref{protocol}), the random variable $n_{1,\rm Z}^{\rm L}$ is introduced in Sec.~\ref{security}, and the random variable $\phi_{1,\rm Z}^{\rm U}$ is defined as follows:
\begin{equation}\label{pher_us}
\phi_{1,\rm Z}^{\rm U}=\frac{\left(n_{1,\rm Z}^{\rm U}+n_{1,\rm X}^{\rm U}\right)\Gamma^+_{\displaystyle{n_{1,\rm X}^{\rm L},\varepsilon}}\left(\displaystyle{\frac{m_{1,\rm X}^{\rm U}}{n_{1,\rm X}^{\rm L}}}\right)-m_{1,\rm X}^{\rm L}}{n_{1,\rm Z}^{\rm L}},
\end{equation}
where the variables $n_{1,\rm Z}^{\rm U}$, $n_{1,\rm X}^{\rm L(U)}$ and $m_{1,\rm X}^{\rm L(U)}$ are also presented in Sec.~\ref{security}. Besides, we recall that $\varepsilon>0$ is the common error probability settled for each individual concentration inequality in Sec.~\ref{security}.

On the other hand, the event where the PE thresholds are not fulfilled by the actual single-photon variables reads
\begin{equation}\label{PE_failure}
\Omega_{\rm PE} = \{ n_{1,\rm Z} \leq{} n_{1,\rm Z}^{\rm th} \} \cup \{ \phi_{1,\rm Z} \geq{} \phi_{1,\rm Z}^{\rm th} \},
\end{equation}
where $n_{1,\rm Z}$ stands for the number of single-photon counts in the sifted key, and $\phi_{1,\rm Z}$ stands for the corresponding single-photon PHER.

Given the definitions of $\Omega_{\rm TEST}$ and $\Omega_{\rm PE}$, bounding the failure probability of the protocol amounts to finding $\epsilon_{\rm PE}>0$ such that $\Pr\left[\Omega_{\rm TEST},\Omega_{\rm PE}\right] \leq \epsilon_{\rm PE}$. This can be easily simplified by noticing that
\begin{equation}
\Pr\left[\Omega_{\rm TEST},\Omega_{\rm PE}\right]\leq{}\Pr\left[n_{1,\rm Z}^{\rm L} \geq n_{1,\rm Z} \cup \phi_{1,\rm Z}^{\rm U} \leq \phi_{1,\rm Z} \right],
\end{equation}
which follows because $A\implies{}B$ implies that $\Pr[A]\leq{}\Pr[B]$ for arbitrary events $A$ and $B$. Now, it is convenient to introduce the event where all the decoy-state bounds of Sec.~\ref{security} hold. Namely,
\begin{equation}\label{Lambda}
\Lambda=\left\{n_{1,\rm Z} \in (n_{1,\rm Z}^{\rm L}, n_{1,\rm Z}^{\rm U}),\ n_{1,\rm X} \in (n_{1,\rm X}^{\rm L}, n_{1,\rm X}^{\rm U}),\ m_{1,\rm X} \in (m_{1,\rm X}^{\rm L}, m_{1,\rm X}^{\rm U})\right\}.
\end{equation}
Making use of $\Lambda$, it follows that
\begin{equation}\label{LTP}
\Pr\left[n_{1,\rm Z}^{\rm L} \geq n_{1,\rm Z} \cup \phi_{1,\rm Z}^{\rm U} \leq \phi_{1,\rm Z} \right]\leq{}\Pr\left[\bar\Lambda\right] + \Pr\left[\phi_{1,\rm Z} \geq \phi_{1,\rm Z}^{\rm U},\Lambda\right],
\end{equation}
where we simply exploit the fact that $\Pr[A]\leq{}\Pr[\bar{B}]+\Pr[A,B]$ (for arbitrary events $A$ and $B$) and further invoke the trivial implication $\Lambda \implies n_{1,\rm Z} > n_{1,\rm Z}^{\rm L}$. Regarding the first term in the right-hand side of Eq.~(\ref{LTP}), carefully counting the errors in Sec.~\ref{security} yields $\Pr\left[\bar\Lambda\right]\leq{}15\varepsilon$, and therefore we conclude
\begin{equation}\label{preliminary_bound}
\Pr\left[\Omega_{\rm TEST},\Omega_{\rm PE}\right]\leq{}15\varepsilon+\Pr\left[\phi_{1,\rm Z} \geq \phi_{1,\rm Z}^{\rm U},\Lambda\right].
\end{equation}
All that remains is to set an upper bound on $\Pr\left[\phi_{1,\rm Z} \geq \phi_{1,\rm Z}^{\rm U},\Lambda\right]$. For this, a critical observation is that
\begin{equation}\label{swapping}
\Lambda\implies{}\phi_{1,\rm Z}^{\rm U}
>q^{\rm th}_{\displaystyle{n_{1,\rm Z} + n_{1,\rm X},n_{1,\rm X},\varepsilon}}\left(\frac{m_{1,\rm X}}{n_{1,\rm X}}\right),
\end{equation}
for the threshold function $q^{\rm th}_{N,n,\epsilon_{\rm PE}}(x)$ defined in Eq.~(\ref{threshold}). Crucially, note that the right-hand side in the inequality of Eq.~(\ref{swapping}) results from ``dropping the superscripts" in the left-hand side (given by Eq.~(\ref{pher_us})), or equivalently, from replacing the decoy random variables $n^{\rm L(U)}_{1,{\rm X(Z)}},n^{\rm L(U)}_{1,{\rm  X(Z)} }$ and $m^{\rm L(U)}_{1,{\rm  X} }$ by the actual single-photon random variables $n_{1,{\rm X(Z)}},n_{1,{\rm  X(Z)} }$ and $m_{1,{\rm  X} }$. Indeed, both the definition of $\phi_{1,\rm Z}^{\rm U}$ and the choice of $\Lambda$ are purely ad-hoc in order to enable Eq.~(\ref{swapping}), because Proposition 2 in the main text does not apply to the decoy random variables but only to the actual single-photon random variables. In summary, because of Eq.~(\ref{swapping}),
\begin{equation}\label{swapping_2}
\Pr\left[\phi_{1,\rm Z} \geq \phi_{1,\rm Z}^{\rm U},\Lambda\right]\leq{}\Pr\left[\phi_{1,\rm Z} \geq{}q^{\rm th}_{\displaystyle{n_{1,\rm Z} + n_{1,\rm X},n_{1,\rm X},\varepsilon}}\left(\frac{m_{1,\rm X}}{n_{1,\rm X}}\right) \right],
\end{equation}
since, again, $A\implies{}B$ implies that $\Pr[A]\leq{}\Pr[B]$. At this point, the only technical obstacle for the application of Proposition 2 in the main text is that this bound presumes a population with fixed size and frequency of ones (errors), and also a fixed test sample size. Therefore, by averaging over population sizes, $n_{1,\rm Z} + n_{1,\rm X}$, test sample sizes, $n_{1,\rm X}$, and total frequencies of errors, $(m_{1,\rm Z}^{\rm ph}+m_{1,\rm X})/(n_{1,\rm Z} + n_{1,\rm X})$ ---where we have introduced $m_{1,\rm Z}^{\rm ph}=\phi_{1,\rm Z}n_{1,\rm Z}$ for clarity---, Proposition 2 in the main text applies for each individual term in the average and, therefore,
\begin{equation}
\Pr\left[\phi_{1,\rm Z} \geq{}q^{\rm th}_{\displaystyle{n_{1,\rm Z} + n_{1,\rm X},n_{1,\rm X},\varepsilon}}\left(\frac{m_{1,\rm X}}{n_{1,\rm X}}\right) \right]\leq{\varepsilon}
\end{equation}
follows. Note that this averaging technique is of the exact same nature as the one used to bound the failure probability of the BBM92 protocol (Proposition 3 in the main text), also deployed in~\cite{Leverrier,Ekert}. Lastly, combining this bound with Eq.~(\ref{swapping_2}) and Eq.~(\ref{preliminary_bound}), we reach $\Pr\left[\Omega_{\rm TEST},\Omega_{\rm PE}\right]\leq{}16\varepsilon$. In conclusion, with the random sampling tool proposed in this work, $\epsilon_{\rm PE}=16\varepsilon$ provides the desired upper bound on the failure probability.
\subsubsection{Serfling and Hush \& Scovel inequalities combined}\label{alternative}
The result in~\cite{Ekert} is devised for the ideal BBM92 protocol, and it is a side contribution of this work to make it compatible with the practical decoy-state BB84 protocol. For this purpose, the alternative PE test we consider is
\begin{equation}
\Omega_{\rm TEST}=\left\{n_{1,\rm Z}^{\rm L} \geq n_{1,\rm Z}^{\rm th,L}, n_{1,\rm Z}^{\rm U} \leq n_{1,\rm Z}^{\rm th,U}, n_{1,\rm X}^{\rm L} \geq n_{1,\rm X}^{\rm th,L}, n_{1,\rm X}^{\rm U} \leq n_{1,\rm X}^{\rm th,U}, m_{1,\rm X}^{\rm U} \leq m_{1,\rm X}^{\rm th,U}\right\},
\end{equation}
where the threshold values $n_{1,\rm Z}^{\rm th,L(U)}$, $n_{1,\rm Z}^{\rm th,L(U)}$ and $m_{1,\rm X}^{\rm th,U}$ are the test-specific protocol inputs (see Sec.~\ref{protocol}). On the other hand, the event where the PE thresholds do not hold is again given by $\Omega_{\rm PE} = \{ n_{1,\rm Z} \leq{} n_{1,\rm Z}^{\rm th} \} \cup \{ \phi_{1,\rm Z} \geq{} \phi_{1,\rm Z}^{\rm th} \}$, recalling that, in this case, $n_{1,\rm Z}^{\rm th}$ and $\phi_{1,\rm Z}^{\rm th}$ are defined in Eq.~(\ref{functions}) in terms of the thresholds above. 

We now introduce the event
\begin{equation}\label{Sigma}
\Sigma=\left\{n_{1,\rm Z} \in \left(n_{1,\rm Z}^{\rm L}, n_{1,\rm Z}^{\rm U}\right),\ n_{1,\rm X} \in \left(n_{1,\rm X}^{\rm L}, n_{1,\rm X}^{\rm U}\right),\ m_{1,\rm X}<m_{1,\rm X}^{\rm U}\right\},
\end{equation}
which matches $\Lambda$ in Eq.~(\ref{Lambda}) exactly except from the fact that it does not contemplate the condition $\{m_{1,\rm X}>m_{1,\rm X}^{\rm L}\}$. In particular, this implies that $\Pr[\bar{\Sigma}]=12\varepsilon$ (composing the relevant errors in Sec.~\ref{failure_decoy}), and therefore $\Pr\left[\Omega_{\rm TEST},\Omega_{\rm PE}\right]\leq{}12\varepsilon + \Pr\left[\Omega_{\rm TEST},\Omega_{\rm PE},\Sigma\right]$. Crucially at this point, the ad hoc definitions of $\Omega_{\rm TEST}$ and $\Sigma$ assure the following trivial implications:
\begin{eqnarray}\label{implications}
&&(\Omega_{\rm TEST},\Sigma)\implies{}\left\{n_{1,\rm Z(X)}\in{}\left(n_{1,\rm Z(X)}^{\rm th,L},n_{1,\rm Z(X)}^{\rm th,U}\right),\ m_{1,\rm X}<m_{1,\rm X}^{\rm th,U}\right\}\implies{}\nonumber \\
&&\left\{e_{1,\rm X}\leq{}e_{1,\rm X}^{\rm th},\ \phi_{1,\rm Z}^{\rm th}\geq{}q^{\rm th}_{\displaystyle{n_{1,\rm Z} + n_{1,\rm X},n_{1,\rm X},\varepsilon}}\left(e_{1,\rm X}^{\rm th}\right)\right\},
\end{eqnarray}
where we have introduced the single-photon bit-error rate $e_{1,\rm X}=m_{1,\rm X}/n_{1,\rm X}$. Namely, the conjunction of $\Omega_{\rm TEST}$ and $\Sigma$ guarantees that the actual single-photon variables ---and not only their decoy-state bounds--- fulfill the thresholds of the PE test (first implication in Eq.~(\ref{implications})), which in turn implies that the random variable $e_{1,\rm X}=m_{1,\rm X}/n_{1,\rm X}$ is upper-bounded by the value $e_{1,\rm X}^{\rm th}=m^{\rm th,U}_{1,\rm X}/n^{\rm th,L}_{1,\rm X}$, and the value $\phi_{1,\rm Z}^{\rm th}=\max_{u,v,\in\mathcal{B}}q_{u,v,\epsilon_{\rm PE}}^{\rm th}(e_{1,\rm X}^{\rm th})$ is equal or larger than the function of random variables $q^{\rm th}_{\displaystyle{n_{1,\rm Z} + n_{1,\rm X},n_{1,\rm X},\varepsilon}}\left(e_{1,\rm X}^{\rm th}\right)$ (second implication in Eq.~(\ref{implications})). The latter holds simply because, if the actual single-photon random variables $n_{1,\rm Z}$ and $n_{1,\rm X}$ fulfill the thresholds of the PE test, the threshold function evaluated in these variables cannot be larger than the threshold function evaluated in the worst-case variables compatible with the thresholds (captivated by the maximization of some arbitrary inputs $u$ and $v$ over $\mathcal{B}$).


In short, we have 
\begin{eqnarray}
&&\Pr\left[\Omega_{\rm TEST},\Omega_{\rm PE},\Sigma\right]\leq{}\Pr\left[n_{1,\rm Z(X)}\in{}\left(n_{1,\rm Z(X)}^{\rm th,L},n_{1,\rm Z(X)}^{\rm th,U}\right),\ m_{1,\rm X}<m_{1,\rm X}^{\rm th,U},\ \phi_{1,\rm Z} \geq{} \phi_{1,\rm Z}^{\rm th}\right]\nonumber \\
&&\leq{}\Pr\left[e_{1,\rm X}\leq{}e_{1,\rm X}^{\rm th},\ \phi_{1,\rm Z}^{\rm th} \geq{} q^{\rm th}_{\displaystyle{n_{1,\rm Z} + n_{1,\rm X},n_{1,\rm X},\varepsilon}}\left(e_{1,\rm X}^{\rm th}\right),\ \phi_{1,\rm Z} \geq{}\ \phi_{1,\rm Z}^{\rm th}\right]\nonumber \\
&&\leq{}\Pr\left[e_{1,\rm X}\leq{}e_{1,\rm X}^{\rm th},\ \phi_{1,\rm Z} \geq{}\ q^{\rm th}_{\displaystyle{n_{1,\rm Z} + n_{1,\rm X},n_{1,\rm X},\varepsilon}}\left(e_{1,\rm X}^{\rm th}\right)\right],
\end{eqnarray}
where in the first inequality we invoke the first implication of Eq.~(\ref{implications}) ---and disregard the event $n_{1,\rm Z}\leq{}n_{1,\rm Z}^{\rm th}$ of $\Omega_{\rm PE}$ because it is incompatible with $n_{1,\rm Z}\in{}\bigl(n_{1,\rm Z}^{\rm th,L},n_{1,\rm Z}^{\rm th,U}\bigr)$---, in the second inequality we invoke the second implication of Eq.~(\ref{implications}), and in the third inequality we make use of the fact that $\left\{\phi_{1,\rm Z}^{\rm th} \geq{} q^{\rm th}_{\displaystyle{n_{1,\rm Z} + n_{1,\rm X},n_{1,\rm X},\varepsilon}}\left(e_{1,\rm X}^{\rm th}\right),\ \phi_{1,\rm Z} \geq{}\ \phi_{1,\rm Z}^{\rm th}\right\}\implies{}\left\{\phi_{1,\rm Z} \geq{}\ q^{\rm th}_{\displaystyle{n_{1,\rm Z} + n_{1,\rm X},n_{1,\rm X},\varepsilon}}\left(e_{1,\rm X}^{\rm th}\right)\right\}$. To finish with, the desired bound follows from an averaging technique, as usual. Particularly, averaging over population sizes ($n_{1,\rm Z} + n_{1,\rm X}$) and test sample sizes ($n_{1,\rm X}$) in this case, it follows that
\begin{equation}
\Pr\left[e_{1,\rm X}\leq{}e_{1,\rm X}^{\rm th},\ \phi_{1,\rm Z} \geq{}\ q^{\rm th}_{\displaystyle{n_{1,\rm Z} + n_{1,\rm X},n_{1,\rm X},\varepsilon}}\left(e_{1,\rm X}^{\rm th}\right)\right]\leq{}\varepsilon
\end{equation}
from the random sampling bound derived in~\cite{Ekert}. Putting it all together, we reach $\Pr\left[\Omega_{\rm TEST},\Omega_{\rm PE}\right]\leq{}13\varepsilon$, and thus we can take $\epsilon_{\rm PE}=13\varepsilon$.
\subsubsection{Serfling inequality}\label{alternative2}
With this tool~\cite{Leverrier,Serfling}, we define the PE test of the protocol equally as in Sec.~\ref{ours}, but taking
\begin{equation}\label{Serfling_worstcase}
\phi_{1,\rm Z}^{\rm U}=\displaystyle{\frac{m_{1,\rm X}^{\rm U}}{n_{1,\rm X}^{\rm L}}}+\sqrt{\frac{(n_{1,\rm Z}^{\rm U}+n_{1,\rm X}^{\rm U})(n_{1,\rm X}^{\rm U}+1)\ln(\varepsilon^{-1})}{2{}n_{1,\rm Z}^{\rm L}\left(n_{1,\rm X}^{\rm L}\right)^2}}.
\end{equation}
Namely, we use the threshold function of Eq.~(\ref{Serfling_threshold}) rather than that of Eq.~(\ref{threshold}). In these circumstances, one can readily prove the failure probability bound $\epsilon_{\rm PE}=13\varepsilon$ reproducing the steps of Sec.~\ref{ours}, but with the ancillary event $\Sigma$ (Eq.~(\ref{Sigma})) rather than $\Lambda$ (Eq.~(\ref{Lambda})). This analogy with Sec.~\ref{ours} comes from the fact that the bound based on Serfling inequality can be trivially stated in the form of Proposition 2 in the main text, just like the bound proposed in this work.

\subsubsection{Exact hypergeometric cumulative mass function}\label{alternative3}
Similarly, the bound based on the Clopper-Pearson method~\cite{Buonaccorsi,Wright} can be stated in the form of Proposition 2 in the main text too. As a consequence, considering the simple PE test of Sec.~\ref{ours}, one can establish that $\epsilon_{\rm PE}=16\epsilon$ is a valid choice following the exact same steps presented there ---in particular, keeping the ancillary event $\Lambda$ of Eq.~(\ref{Lambda})---, but taking
\begin{equation}
\phi_{1,\rm Z}^{\rm U}=\frac{\left(n_{1,\rm Z}^{\rm U}+n_{1,\rm X}^{\rm U}\right){\mathcal{CP}}^+_{\displaystyle{n_{1,\rm Z}^{\rm U}+n_{1,\rm X}^{\rm U},n_{1,\rm X}^{\rm L},\varepsilon}}\left(\displaystyle{\frac{m_{1,\rm X}^{\rm U}}{n_{1,\rm X}^{\rm L}}}\right)-m_{1,\rm X}^{\rm L}}{n_{1,\rm Z}^{\rm L}}.
\end{equation}
\subsection{Channel model}\label{channel model}
Here, we describe the detector and channel model deployed in the decoy-state BB84 simulations of the main text. The expected detection rate of a signal of intensity $k$ is given by
\begin{equation}
D_k =1-(1-2p_{\rm d})\exp(-\eta k),
\end{equation}
and the probability of having a bit error for intensity $k$ is given by
\begin{equation}
e_k=p_{\rm d}+e_{\rm mis}\left[1-\exp(-\eta k) \right],
\end{equation}
where $p_d$ is the dark count probability of Bob's detectors,  $e_{\rm mis}$ is the misalignment error rate of the system, and $\eta $ is the overall system efficiency. Namely, $\eta=10^{-\lambda/10}$, $\lambda$ denoting the total loss in dB.

This model is used to select the threshold values of the PE test. To be precise, the thresholds are set to the expectation values of the underlying variables according to the channel model. As an example, within the PE test $\Omega_{\rm TEST} = \{n_{1,\rm Z}^{\rm L} \geq n_{1,\rm Z}^{\rm th},\ \phi_{1,\rm Z}^{\rm U} \leq \phi_{1,\rm Z}^{\rm th}\}$, $n_{1,\rm Z}^{\rm th}$ is set to the decoy-state bound of Eq.~(\ref{singlow}), upon replacement of $n_{{\rm Z},k}$ by its expectation according to the channel model:
\begin{equation}\label{expectation1}
E\left[n_{{\rm Z},k}\right]=N_{\rm Z}\frac{p_k D_k}{\sum_j p_j D_j}.
\end{equation}
Similarly, $\phi_{1,\rm Z}^{\rm th}$ is set to the threshold function under consideration, upon replacement of all the observables ($n_{{\rm Z},k}$, $n_{{\rm X},k}$ and $m_{{\rm X},k}$) by their expectations according to the channel model. Particularly,
\begin{equation}\label{expectation2}
E\left[m_{{\rm X},k}\right]=N_{\rm X}\frac{p_k e_k}{\sum_j p_j D_j},
\end{equation}
and $E[n_{{\rm X},k}]$ is obtained simply replacing $Z$ by $X$ in Eq.~(\ref{expectation1}).

We remark that the above choices of the thresholds in the PE test are over-optimistic, as they would incur in large abortion probabilities even if the channel behaves according to the considered model~\cite{robustness}. A more sensitive approach consists of making robustness considerations to incorporate statistical fluctuations in the thresholds. Ultimately, however, in a real experiment, the thresholds would be chosen via careful channel monitoring.\\

\section*{Appendix III: Computing confidence intervals from the additive Chernoff bounds}\label{III}
This section includes the technical derivations underlying Proposition 1 in the main text. In the first place, we provide a relaxation of the upper additive Chernoff bound~\cite{Chernoff}.
\begin{proposition}
Let $\hat{p}$ be the average of $n$ independent Bernoulli variables, with expected value $\mathbb{E}\left(\hat{p}\right)=p$. Then, for all $z\in[p,1]$, $\Pr[\hat{p}\geq{}z]\leq{e^{-n{}D(z,p)}}$ for $D(z,p)=9(z-p)^2/2(z+2p)(3-z-2p)$.
\end{proposition}
\begin{proof}
We recall that the original bound~\cite{Chernoff} is analogous to the above proposition but replacing $D(z,p)$ by the Kullback-Leibler divergence,
\begin{equation}\label{divergence}
D(z||p)=z\ln\left(\frac{z}{p}\right)+(1-z)\ln\left(\frac{1-z}{1-p}\right).
\end{equation}
Therefore, it suffices to show that $D(z||p)\geq{}D(z,p)$ for all $z\in[p,1]$. For this purpose, we build on the logarithmic inequality~\cite{Topsoe}
\begin{equation}\label{ineq}
\ln(1+x)\geq{}\frac{x(6+5x)}{2(1+x)(3+x)}.
\end{equation}
Although this inequality is only established for $-1<x\leq{}0$ in~\cite{Topsoe}, it actually holds for all $x>-1$, which is required in this proof. To see this, we note that the difference
\begin{equation}
h(x)=\ln(1+x)-\frac{x(6+5x)}{2(1+x)(3+x)}
\end{equation}
is non-negative for all $x>-1$. This is shown by the fact that the derivative of $h(x)$ is given by
\begin{equation}
h'(x)=\frac{x^3}{(1+x)^2(3+x)^2},
\end{equation}
revealing that $h(x)$ reaches a global minimum at $x=0$, given by $h(0)=0$. Plugging the inequality in each term of the divergence yields
\begin{equation}
D(z||p) \geq z \frac{(\frac{z}{p}-1)(6+5(\frac{z}{p}-1))}{2(1+(\frac{z}{p}-1))(3+(\frac{z}{p}-1))}+(1-z)\frac{(\frac{1-z}{1-p}-1)(6+5(\frac{1-z}{1-p}-1))}{2(1+(\frac{1-z}{1-p}-1))(3+(\frac{1-z}{1-p}-1))},
\end{equation}
which, after some straightforward manipulation, leads to
\begin{equation}
D(z||p) \geq{} \frac{(z-p)(p+5z)}{2(z+2p)}+ \frac{(p-z)(6-5z-p)}{2(3-z-2p)}.
\end{equation}
Lastly, combining both summands in a single fraction and simplifying yields
\begin{equation}
D(z||p) \geq \frac{9(z-p)^2}{2(z+2p)(3-z-2p) }.
\end{equation}
Hence, the desired relaxation follows.
\end{proof}
Coming next, we reformulate Proposition 1 above to determine the upper bound on $\hat{p}$ that matches any desired error probability $\epsilon$.
\begin{proposition}
Let $\hat{p}$ be the average of $n$ independent Bernoulli variables, with expected value $\mathbb{E}\left(\hat{p}\right)=p$, and let $\epsilon>0$. If $p\leq{}(1-2\kappa_{n,\epsilon})/(1+4\kappa_{n,\epsilon})$ for $\kappa_{n,\epsilon}=(2/9n)\ln(1/\epsilon)$, then $\Pr[\hat{p}\geq{}z_{n,\epsilon}(p)]\leq{\epsilon}$ with
\begin{equation}\label{piecewise_1}
z_{n,\epsilon}(p)=\displaystyle{\frac{-b_{n,\epsilon}(p)+\sqrt{b_{n,\epsilon}(p)^{2}-4a_{n,\epsilon}c_{n,\epsilon}(p)}}{2a_{n,\epsilon}}},
\end{equation}
where $a_{n,\epsilon}=1+\kappa_{n,\epsilon}$, $b_{n,\epsilon}(p)=-2p-\kappa_{n,\epsilon}(3-4p)$ and $c_{n,\epsilon}(p)=p^2-\kappa_{n,\epsilon}{}p(6-4p)$.
\end{proposition}
\begin{proof}
According to Proposition 1, for all $z\in[p,1]$, $\Pr[\hat{p}\geq{}z]\leq{\epsilon}$ holds if $\epsilon=\exp\{-n{}D(z,p)\}$. Let us consider the function $g_{n,p}(z)=\exp\{-n{}D(z,p)\}$ that maps $z$'s to $\epsilon$'s for arbitrary $n$ and $p$. Since $g_{n,p}(z)$ is decreasing and thus injective in $[p,1]$, it admits an inverse in $g_{n,p}\left([p,1]\right)=\bigl[g_{n,p}(1),1\bigr]$. Namely, for all $\epsilon\geq{}g_{n,p}(1)$, $\Pr[\hat{p}\geq{}z]\leq{\epsilon}$ holds if $z=g^{-1}_{n,p}(\epsilon)$, which is the only $z\in[p,1]$ such that $g_{n,p}(z)=\epsilon$. This is a quadratic equation $a_{n,\epsilon}z^{2}+b_{n,\epsilon}(p){}z+c_{n,\epsilon}(p)=0$ with the coefficients prescribed in the statement, and one can readily show that only the larger root lies in $[p,1]$. Lastly, the claim follows by observing that $\epsilon\geq{}g_{n,p}(1)\Leftrightarrow{}p\leq{}(1-2\kappa_{n,\epsilon})/(1+4\kappa_{n,\epsilon})$.
\end{proof}
Let us now use Proposition 2 to infer a confidence bound for the underlying parameter $p$~\cite{Neyman1,Neyman2}. Analytical inspection of ${dz}_{n,\epsilon}/dp$ shows that ${z}_{n,\epsilon}(p)$ is monotonically increasing in the interval $[0,(1-2\kappa_{n,\epsilon})/(1+4\kappa_{n,\epsilon})]$, and thus it admits an inverse $z^{-1}_{n,\epsilon}$ in ${z}_{n,\epsilon}([0,(1-2\kappa_{n,\epsilon})/(1+4\kappa_{n,\epsilon})])=\left[3\kappa_{n,\epsilon}/(1+\kappa_{n,\epsilon}),1\right]=:I_{n,\epsilon}$. By definition, $p=z^{-1}_{n,\epsilon}(x)$ fulfills $z_{n,\epsilon}(p)=x$. Solving this quadratic equation for $p$ and choosing the relevant root we conclude that, for all $x\in{}I_{n,\epsilon}$,
\begin{equation}
z^{-1}_{n,\epsilon}(x)=\frac{1}{1+4\kappa_{n,\epsilon}}\left(3\kappa_{n,\epsilon}+(1-2\kappa_{n,\epsilon})x-3\sqrt{\kappa_{n,\epsilon}\left(\kappa_{n,\epsilon}+x-x^{2}\right)}\right).
\end{equation} 
Essentially, $z^{-1}_{n,\epsilon}(\hat{p})$ is the statistic we use to construct a confidence bound on $p$, as shown next.
\begin{proposition}
Let $\hat{p}$ be the average of $n$ independent Bernoulli variables, with expected value $\mathbb{E}\left[\hat{p}\right]=p$. Then, for all $\epsilon>0$, $\Pr[\Gamma^-_{n,\epsilon}(\hat{p})\geq{}p]\leq{\epsilon}$, where $\Gamma^-_{n,\epsilon}(\hat{p})=z^{-1}_{n,\epsilon}(\hat{p})$ if $\hat{p}\in{}I_{n,\epsilon}$, and $\Gamma^-_{n,\epsilon}(\hat{p})=-\epsilon$ otherwise.
\end{proposition}
\begin{proof}
For all $\hat{p}$, $\Gamma^-_{n,\epsilon}(\hat{p})\leq{}\Gamma^-_{n,\epsilon}(1)=(1-2\kappa_{n,\epsilon})/(1+4\kappa_{n,\epsilon})$. Hence, the claim holds if $(1-2\kappa_{n,\epsilon})/(1+4\kappa_{n,\epsilon})<p\leq{}1$. Let us now consider $0\leq{}p\leq{}(1-2\kappa_{n,\epsilon})/(1+4\kappa_{n,\epsilon})$. We have 
$\Pr\left[\Gamma^-_{n,\epsilon}(\hat{p})\geq{}p\right]=\Pr\left[\hat{p}\in{}I_{n,\epsilon},\Gamma^-_{n,\epsilon}(\hat{p})\geq{}p\right]=\Pr\left[\hat{p}\in{}I_{n,\epsilon},z^{-1}_{n,\epsilon}(\hat{p})\geq{}p\right]\leq{}\Pr\left[\hat{p}\in{}I_{n,\epsilon},\hat{p}\geq{}z_{n,\epsilon}(p)\right]\leq{}\Pr\left[\hat{p}\geq{}z_{n,\epsilon}(p)\right]\leq{}\epsilon.$ Here, the first equality follows because $\hat{p}\notin{}I_{n,\epsilon}\implies{}\Gamma^-_{n,\epsilon}(\hat{p})=-\epsilon<p$~\cite{subtlety}, the second equality follows because $\hat{p}\in{}I_{n,\epsilon}\implies\Gamma^-_{n,\epsilon}(\hat{p})=z^{-1}_{n,\epsilon}(\hat{p})$, the first inequality follows because $\{\hat{p}\in{}I_{n,\epsilon},z^{-1}_{n,\epsilon}(\hat{p})\geq{}p\}\implies{}\{\hat{p}\in{}I_{n,\epsilon},\hat{p}\geq{}z_{n,\epsilon}(p)\}$ ---as $z^{-1}_{n,\epsilon}(\hat{p})$ is monotonically increasing in $I_{n,\epsilon}$---, the second inequality follows because $\{\hat{p}\in{}I_{n,\epsilon},\hat{p}\geq{}z_{n,\epsilon}(p)\}\implies{}\{\hat{p}\geq{}z_{n,\epsilon}(p)\}$, and the final inequality follows from Proposition 2 above.
\end{proof}

We remark that, ultimately, Propositions 1 to 3 originate from the upper additive Chernoff bound. Importantly, for each of them, there exists an analogous statement arising from the lower additive Chernoff bound instead. In particular, a confidence bound complementary to that of Proposition 3 follows.
\begin{proposition}
Let $\hat{p}$ be the average of $n$ independent Bernoulli variables, with expected value $\mathbb{E}\left(\hat{p}\right)=p$. Then, for all $\epsilon>0$, $\Pr[\Gamma^+_{n,\epsilon}(\hat{p})\leq{}p]\leq{\epsilon}$, where
\begin{equation}\label{upper}
\Gamma^+_{n,\epsilon}(\hat{p})=\frac{1}{1+4\kappa_{n,\epsilon}}\left(3\kappa_{n,\epsilon}+(1-2\kappa_{n,\epsilon})\hat{p}+3\sqrt{\kappa_{n,\epsilon}\left(\kappa_{n,\epsilon}+\hat{p}-\hat{p}^{2}\right)}\right)
\end{equation}
if $\hat{p}\in[0,(1-2\kappa_{n,\epsilon})/(1+\kappa_{n,\epsilon})]$, and $\Gamma^+_{n,\epsilon}(\hat{p})=1+\epsilon$ otherwise.
\end{proposition}
\begin{proof}
The proof follows similar lines as that of Proposition 3 above, but using the lower additive Chernoff bound as the starting point.
\end{proof}
Proposition 1 in the main text is the conjunction of Propositions 3 and 4 presented here.\\

\section*{Appendix IV: Optimal confidence bounds for the average parameter of a set of independent Bernoulli variables}\label{IV}
Here, we present the tightest monotonic one-sided confidence intervals for the average parameter of a set of independent Bernoulli variables, originally derived in~\cite{BancalSekatski}.

\begin{theorem}
Let $\hat{p}$ be the average of $n$ independent Bernoulli variables, with expected value $\mathbb{E}\left[\hat{p}\right]=p$. Then, for all $\epsilon\in(0,1/4]$, we have that $\Pr[\mathcal{F}^{-}_{n,\epsilon}(\hat{p})>p]\leq{\epsilon}$ for
\begin{equation}
\mathcal{F}^{-}_{n,\epsilon}(\hat{p})=\left\{
\begin{array}{ll}
0 & \mathrm{if}\hspace{.2cm}n\hat{p}=0, \\
\hat{p}-\displaystyle{\frac{1-\epsilon}{n\left[1-\epsilon^{*}(n\hat{p},n)\right]}} & \mathrm{if}\hspace{.2cm}n\hat{p}>0,\hspace{.2cm}\epsilon^{*}(n\hat{p},n)\leq{}\epsilon\leq{}1, \\
I^{-1}_{\epsilon}(n\hat{p},n-n\hat{p}+1) & \mathrm{if}\hspace{.2cm}n\hat{p}>0,\hspace{.2cm}0\leq{}\epsilon\leq{}\epsilon^{*}(n\hat{p},n),
\end{array} 
\right.
\end{equation}
where $\epsilon^{*}(n\hat{p},n)=I_{(n\hat{p}-1)/n}(n\hat{p},n-n\hat{p}+1)$ and $I^{-1}_{\epsilon}(a,b)$ denotes the inverse regularized incomplete beta function, such that $I_{I^{-1}_{\epsilon}(a,b)}(a,b)=\epsilon$ for
\begin{equation}
I_{x}(a,b)=\frac{\displaystyle{\int_{0}^{x}t^{a-1}(1-t)^{b-1}dt}}{\displaystyle{\int_{0}^{1}t^{a-1}(1-t)^{b-1}dt}}.
\end{equation}
Complementarily, for all $\epsilon\in(0,1/4]$, we have that $\Pr[\mathcal{F}^{+}_{n,\epsilon}(\hat{p})<p]\leq{\epsilon}$ for $\mathcal{F}^{+}_{n,\epsilon}(\hat{p})=1-\mathcal{F}^{-}_{n,\epsilon}(1-\hat{p})$.
\end{theorem}
To finish with, two observations are in order. On the one hand, note that for any $\Delta>0$, $\left\{\mathcal{F}^{-}_{n,\epsilon}(\hat{p})\geq{}p+\Delta\right\}\implies{}\left\{\mathcal{F}^{-}_{n,\epsilon}(\hat{p})>p\right\}$, such that $\Pr[\mathcal{F}^{-}_{n,\epsilon}(\hat{p})\geq{}p+\Delta]=\Pr[\left(\mathcal{F}^{-}_{n,\epsilon}(\hat{p})-\Delta\right)\geq{}p]\leq{\epsilon}$ in virtue of Theorem 1. In other words, by decreasing $\mathcal{F}^{-}_{n,\epsilon}(\hat{p})$ in any nonzero amount $\Delta$, we reach a looser lower bound that trivially sticks to the sign convention used in this work, \emph{i.e.} contemplating the equality sign in the probability (rather than the strict inequality considered in Theorem 1). Similarly, increasing $\mathcal{F}^{+}_{n,\epsilon}(\hat{p})$ in any nonzero amount allows to meet the desired criterion for the upper bound.

On another note, binomial-like confidence bounds slightly simpler than those of Theorem 1 are provided in~\cite{BancalSekatski2}, which only entail a trivial loss of optimality imperceptible with the sample sizes characteristic of QKD applications.
\section*{Appendix V: Comparative survey of hypergeometric tail bounds}\label{V}
It is known that the pointwise probabilities of the distributions $\mathrm{Hypergeometric}(N,K,n)$ and $\mathrm{Binom}(n,K/N)$ differ at most by $(n-1)/(N-1)$ for any possible outcome~\cite{Holmes}. This being the case, the smaller the sampling fraction $n/N$ is, the tighter the Chernoff bound becomes in the hypergeometric setting. In this section, we go beyond this asymptotic argument and compare our relaxed Chernoff bound with selected concentration inequalities in the literature, characterizing the parameter regimes where our bound dominates. Notably, few concentration results are known for the hypergeometric distribution, and we only consider those susceptible of deriving analytical confidence upper bounds (this leads to the exclusion of e.g. a bound presented in~\cite{Hayashi_1} and various results from~\cite{Greene-thesis,Greene}). Overall, we select six different bounds that meet this criterion: Serfling~\cite{Serfling} and Hush \& Scovel~\cite{Hush} inequalities ---already explored in this work--- and four Bernstein-type~\cite{Bernstein,Bennett} inequalities respectively due to Greene \& Wellner~\cite{Greene}, Bardenet \& Maillard~\cite{Bardenet}, Goldstein \& Islak~\cite{Goldstein} and Chatterjee~\cite{Chatterjee}.

As in the main text, we shall consider $\hat{X}\sim{}\mathrm{Hypergeometric}(N,K,n)$ and introduce $\hat{p}=\hat{X}/n$, such that $\mathbb{E}\left[\hat{p}\right]=p$ with $p=K/N$. Also, we assume $p\leq{}1/2$ and $n/N\leq{}1/2$, which suffices to illustrate the advantage of our relaxed Chernoff bound. In fact, for a fair comparison in the context $n/N>1/2$, a bound complementary to ours should be deployed, obtainable by exploiting elementary symmetries of the hypergeometric distribution (see e.g.~\cite{blog}).\\

\subsection{Hoeffding-type bounds~\cite{Serfling,Hush}}\label{1st_type}
We refer to Serfling and Hush \& Scovel inequalities as Hoeffding-type because they rely on the range of the random variables---\emph{i.e.} the interval $[0,1]$ for our purposes---, in contrast to Bernstein-type inequalities that rely on the variance of the population. Since Hush \& Scovel is not unconditionally tighter than Serfling, we draw separate comparisons of our bound with the two inequalities, starting with Serfling~\cite{refinement}.

\subsubsection{Serfling inequality}
For all $\delta\in[0,p]$, $\Pr\left[\hat{p}\leq{}p-\delta\right]\leq{}\exp\left[-2Nn\delta^{2}/(N-n+1)\right]$ according to Serfling~\cite{refinement2}, and $\Pr\left[\hat{p}\leq{}p-\delta\right]\leq{}\exp\left[-nD(p-\delta,p)\right]$ according to our relaxed Chernoff bound, where $D(x,y)$ denotes the relaxation of the Kullback-Leibler divergence. Comparing the exponents $\nu_{\rm Serfling}=2Nn\delta^{2}/(N-n+1)$ and $\nu_{\rm Chernoff}=nD(p-\delta,p)$, one can easily show that
\begin{equation}\label{comparison}
\frac{\nu_{\rm Chernoff}}{\nu_{\rm Serfling}}=\frac{N-n+1}{N}\frac{1}{4w(1-w)}
\end{equation}
for $w:={}p-\delta/3\in[2p/3,p]$. Setting $f_{n}^{*}=(n-1)/N$, it follows from Eq.~(\ref{comparison}) that $\nu_{\rm Chernoff}/\nu_{\rm Serfling}<1$ ---\emph{i.e.} Serfling outperforms Chernoff--- when
\begin{equation}\label{serfling_dominance}
f_{n}^{*}>1-4w(1-w).
\end{equation}

This condition is extremely restrictive for applications with small sampling fractions. In particular, it is untenable in the random sampling problem of QKD. To see this, note that $w(1-w)\leq{}p(1-p)$ for $w\in{}[2p/3,p]$ and $p\leq{}1/2$, such that a necessary condition to fulfill Eq.~(\ref{serfling_dominance}) is given by $f_{n}^{*}>1-4p(1-p)$. If we conservatively assume that $p<10\%$ in the context of QKD (typical values are in fact much lower than $10\%$), this necessary condition translates into $f_{n}^{*}>64\%$.

\subsubsection{Hush \& Scovel inequality}
Let us now consider Hush \& Scovel inequality. Hush \& Scovel report a relevant advantage over Serfling if $p<\min(n/N,1-n/N)$ or $p>\max(n/N,1-n/N)$, in which case the inequality reads $\Pr\left[\hat{p}\leq{}p-\delta\right]\leq{}\exp\left[-2\alpha_{N,p}(\delta^2{}n^2-1)\right]$ for all $\delta\in[0,p]$ and $\alpha_{N,p}=(N+2)/(Np+1)(N-Np+1)$~\cite{casuistry}. The bound becomes trivial for $\delta{}n\leq{}1$. Hence, considering $\delta{}n>1$ and denoting the exponent by $\nu_{\rm Hush}$, we have that
\begin{equation}\label{comparison2}
\frac{\nu_{\rm Chernoff}}{\nu_{\rm Hush}}=\frac{(Np+1)(N-Np+1)}{(N+2)n}\left(1-\frac{1}{n^2\delta^2}\right)^{-1}\frac{1}{4w(1-w)}\geq{}\frac{(Np+1)(N-Np+1)}{(N+2)n}\frac{1}{4w(1-w)}\approx{}\frac{N}{4n}\frac{p(1-p)}{w(1-w)},
\end{equation}
where the approximation is tight as long as $Np\gg{}1$, which holds for the vast majority of applications. Now, using again that $w(1-w)\leq{}p(1-p)$ for $w\in{}[2p/3,p]$ and $p\leq{}1/2$, it follows that
\begin{equation}\label{1/4}
\frac{\nu_{\rm Chernoff}}{\nu_{\rm Hush}}\gtrapprox{}\frac{N}{4n}
\end{equation}
according to Eq.~(\ref{comparison2}), meaning that sampling fractions above $1/4$ are required for Hush \& Scovel to dominate. We further explore this potential advantage in Sec.~\ref{VI}, in the specific context of QKD. We remark, however, that such large sampling fractions are only acceptable with critically small block sizes, in which case it is important to stress that the numerical tools developed in this work provide a significant edge anyway.

\subsection{Bernstein-type bounds~\cite{Greene,Bardenet}}\label{2nd_type}
Greene \& Wellner~\cite{Greene}, Bardenet \& Maillard~\cite{Bardenet}, Goldstein \& Islak~\cite{Goldstein} and Chatterjee~\cite{Chatterjee} provide Bernstein-type bounds~\cite{Bernstein,Bennett} for sampling without replacement, and direct comparison reveals that Corollary 1 in Greene \& Wellner~\cite{Greene} always yields a tighter exponent than the other three. Notwithstanding, according to the authors, this corollary is subject to the constraint $n/N\leq{}\min(p,1-p)$, which in principle prevents us from using it for the purpose of statistical inference. Here, however, we show that the corollary holds without the constraint as well, and an elaborate argument on this is given in Sec.~\ref{greene_detail}. In fact, the argument there refers to a statement complementary to the corollary, because it is the lower tail rather than the upper tail that concerns us here. This statement reads
\begin{equation}\label{bernstein}
\Pr\left[\hat{p}\leq{}p-\delta\right]\leq{}\exp\left[-\frac{n\delta^{2}/2}{\sigma^{2}(1-f_{n})+\delta/3}\right]
\end{equation}
for $\sigma^{2}=p(1-p)$, $f_{n}=(n-1)/(N-1)$ and all $\delta\in[0,p]$. Denoting the exponent in Eq.~(\ref{bernstein}) as $\nu_{\rm Greene}$, the request $\nu_{\rm Chernoff}/\nu_{\rm Greene}>1$ can be recasted as
\begin{equation}\label{critical}
f_{n}<\frac{2}{3}\frac{\delta}{p}+\frac{1}{9}\frac{\delta^2}{p(1-p)},
\end{equation}
which gives a sufficient condition for the dominance of our bound over all four inequalities mentioned above. Roughly speaking, for a given $p$, the advantage of our bound is conditional on the ratio $f_{n}/\delta$, with lower ratios favoring our bound and larger ratios favoring the considered Bernstein-type inequalities.
The reader is referred to Sec.~\ref{VI} for a detailed assessment of Eq.~(\ref{bernstein}) in the context of QKD.

\section*{Appendix VI: the case for Hush \& Scovel and Greene \& Wellner inequalities}\label{VI}
In this final section, we evaluate the performance of Hush \& Scovel and Greene \& Wellner inequalities in the specific context of QKD. To this end, we use these tools to address the failure probability estimation of the QKD protocols of Sec.~\ref{I} and Sec.~\ref{II}, and complement this with QKD simulations matching Figure 2 in the main text.

\subsection{Hush \& Scovel}\label{hush_detail}
Let us address Hush \& Scovel inequality first. We stick to the natural procedure that leads to Proposition 2 in the main text, according to which a threshold function for the random sampling problem of QKD follows given a confidence upper bound for the hypergeometric population parameter. The derivation of this confidence bound is formally identical to the bulk of Sec.~\ref{III}. To begin with, we conveniently recast the lower tail inequality of Hush \& Scovel~\cite{Hush}.
\begin{proposition}
Let $\hat{X}\sim{}Hypergeometric(N,K,n)$ and $\hat{p}=\hat{X}/n$, such that $\mathbb{E}\left[\hat{p}\right]=p$ with $p=K/N$. Then, for all $z\in[0,p]$, $\Pr[\hat{p}\leq{}z]\leq{}\exp\left\{-2\alpha_{N,p}\left[n^{2}(p-z)^{2}-1\right]\right\}$ for $\alpha_{N,p}=(N+2)/(Np+1)(N-Np+1)$.
\end{proposition}
Next, we establish the upper bound on $\hat{p}$ that matches any desired error probability $\epsilon$.
\begin{proposition}
Let $\hat{X}\sim{}Hypergeometric(N,K,n)$ and $\hat{p}=\hat{X}/n$, such that $\mathbb{E}\left[\hat{p}\right]=p$ with $p=K/N$. Also, let $\epsilon>0$, $\tau_{N,\epsilon}=\ln(1/\epsilon)\bigl/{2(N+2)}$ and $p^{+}_{N,n,\epsilon}=\left\{\tau_{N,\epsilon}N^{2}+\sqrt{\tau^{2}_{N,\epsilon}N^{4}+4\left(n^{2}+\tau_{N,\epsilon}N^{2}\right)\Bigl[1+\tau_{N,\epsilon}(N+1)\Bigr]}\right\}\Bigl/{2\left(n^{2}+\tau_{N,\epsilon}N^{2}\right)}$. If $p\geq{}p^{+}_{N,n,\epsilon}$, then $\Pr[\hat{p}\leq{}z_{N,n,\epsilon}(p)]\leq{}\epsilon$ for
\begin{equation}\label{z_hush}
z_{N,n,\epsilon}(p)=p-\frac{1}{n}\sqrt{1+\tau_{N,\epsilon}(Np+1)(N-Np+1)}.
\end{equation}
\end{proposition}
\begin{proof}
The function $g_{N,n,p}$ that maps $z$'s to $\epsilon$'s is increasing and thus invertible in its codomain, $[g_{N,n,p}(0),g_{N,n,p}(p)]$, where $\epsilon\geq{}g_{N,n,p}(0)\Leftrightarrow{}p\geq{}p^{+}_{N,n,\epsilon}$ for the prescribed constant $p^{+}_{N,n,\epsilon}$, and $\epsilon\leq{}g_{N,n,p}(p)$ is unrestrictive. The inverse $g_{N,n,p}^{-1}(\epsilon)=:z_{N,n,\epsilon}(p)$ is the only $z\in[0,p]$ fulfilling $g_{N,n,p}(z)=\epsilon$, which is solved explicitly to yield Eq.~(\ref{z_hush}).
\end{proof}
The study of $dz_{N,n,\epsilon}/dp$ shows that ${z}_{N,n,\epsilon}(p)$ increases in the considered interval $\bigl[p^{+}_{N,n,\epsilon},1\bigr]$, thus admitting an inverse $z^{-1}_{N,n,\epsilon}$ in $z_{N,n,\epsilon}\bigl([p^{+}_{N,n,\epsilon},1]\bigr)=\bigl[0,1-\sqrt{1+\tau_{N,\epsilon}(N+1)}\bigl/{}n\bigr]=:I_{N,n,\epsilon}$. By definition, $p=z^{-1}_{N,n,\epsilon}(x)$ fulfills $z_{N,n,\epsilon}(p)=x$, such that solving this quadratic equation for $p$ and choosing the relevant root yields
\begin{equation}
z^{-1}_{N,n,\epsilon}(x)=\frac{\tau_{N,\epsilon}N^{2}+2n^{2}x+\sqrt{\tau^{2}_{N,\epsilon}N^{2}(N+2)^{2}+4\tau_{N,\epsilon}\Bigl[N^{2}n^{2}x(1-x)+N^{2}+(N+1)n^{2}\Bigr]+4n^{2}}}{2\left(n^{2}+\tau_{N,\epsilon}N^{2}\right)}
\end{equation}
for all $x\in{}I_{N,n,\epsilon}$. As shown next, the statistic $z_{N,n,\epsilon}^{-1}(\hat{p})$ provides the desired confidence bound on $p$.
\begin{proposition}
Let $\hat{X}\sim{}Hypergeometric(N,K,n)$ and $\hat{p}=\hat{X}/n$, such that $\mathbb{E}\left[\hat{p}\right]=p$ with $p=K/N$. Then, for all $\epsilon>0$, $\Pr[\Gamma^+_{N,n,\epsilon}(\hat{p})\leq{}p]\leq{\epsilon}$, where $\Gamma^{+}_{N,n,\epsilon}(\hat{p})=z^{-1}_{N,n,\epsilon}(\hat{p})$ if $\hat{p}\in{}I_{N,n,\epsilon}$, and $\Gamma^+_{N,n,\epsilon}(\hat{p})=1+\epsilon$ otherwise.
\end{proposition}
\begin{proof}
If $p<p^{+}_{N,n,\epsilon}$, the claim holds because $\Gamma^+_{N,n,\epsilon}(\hat{p})\geq{}\Gamma^+_{N,n,\epsilon}(0)=z^{-1}_{N,n,\epsilon}(0)=p_{+}>p$. If $p\geq{}p^{+}_{N,n,\epsilon}$, we have that 
$\Pr\left[\Gamma^+_{N,n,\epsilon}(\hat{p})\leq{}p\right]=\Pr\left[\hat{p}\in{}I_{N,n,\epsilon},\Gamma^+_{N,n,\epsilon}(\hat{p})\leq{}p\right]=\Pr\left[\hat{p}\in{}I_{N,n,\epsilon},z^{-1}_{N,n,\epsilon}(\hat{p})\leq{}p\right]\leq{}\Pr\left[\hat{p}\in{}I_{N,n,\epsilon},\hat{p}\leq{}z_{N,n,\epsilon}(p)\right]\leq{}\Pr\left[\hat{p}\leq{}z_{N,n,\epsilon}(p)\right]\leq{}\epsilon.$ The first equality follows because $\hat{p}\notin{}I_{N,n,\epsilon}\implies{}\Gamma^+_{N,n,\epsilon}(\hat{p})
>p$, and the second one follows from the definition of $\Gamma^+_{N,n,\epsilon}(\hat{p})$. The first inequality follows because $z^{-1}_{N,n,\epsilon}(\hat{p})$ increases in $I_{N,n,\epsilon}$, the second one follows because $\{\hat{p}\in{}I_{N,n,\epsilon},\hat{p}\leq{}z_{N,n,\epsilon}(p)\}\implies{}\{\hat{p}\leq{}z_{N,n,\epsilon}(p)\}$ and the last one follows from Proposition 6.
\end{proof}
In virtue of Proposition 7, Propositions 2 and 3 in the main text hold as well~\cite{slope} for the Hush \& Scovel threshold function
\begin{equation}\label{hush_th}
q^{\rm th}_{N,n,\epsilon}(x)=\frac{N\Gamma^+_{N,n,\epsilon}(x)-nx}{N-n}.
\end{equation}
We recall that this function provides the failure probability estimation of the BBM92 protocol. In particular, the confidence error ---\emph{i.e.} the input $\epsilon$ in the threshold function--- is the necessary upper bound on the failure probability, required to quantify the secrecy parameter. In the decoy-state BB84 protocol, however, this confidence error is just one contribution to the failure probability, to be combined with the statistical errors of the decoy-state bounds. In this regard, if, for simplicity, we assume again a common error probability $\varepsilon$ per statistical bound, one can show that $\Pr\left[\Omega_{\rm TEST},\Omega_{\rm PE}\right]\leq{}16\varepsilon$ for the simple PE test of Eq.~(\ref{simpler}). For this, it suffices to follow the exact same steps of Sec.~\ref{ours}, but taking
\begin{equation}\label{pher_HS}
\phi_{1,\rm Z}^{\rm U}=\frac{\left(n_{1,\rm Z}^{\rm U}+n_{1,\rm X}^{\rm U}\right){\Gamma}^+_{\displaystyle{n_{1,\rm Z}^{\rm U}+n_{1,\rm X}^{\rm U},n_{1,\rm X}^{\rm L},\varepsilon}}\left(\displaystyle{\frac{m_{1,\rm X}^{\rm U}}{n_{1,\rm X}^{\rm L}}}\right)-m_{1,\rm X}^{\rm L}}{n_{1,\rm Z}^{\rm L}}
\end{equation}
rather than Eq.~(\ref{pher_us}).

\subsection{Greene \& Wellner}\label{greene_detail}
Let us now consider the Greene \& Wellner inequality of Eq.~(\ref{bernstein}). In the first place, we show that this equation follows without resorting to the constraint $n/N\leq{}\min(p,1-p)$, where we recall that the parameters correspond to a random variable $\hat{X}\sim{}Hypergeometric(N,K,n)$ with $p=K/N$. Matching Sec.~\ref{V}, we assume $p\leq{}1/2$, in which case the constraint reads $n/N\leq{}p$. From Theorem 1 in~\cite{Hui} (see also Proposition 1 in the earlier work~\cite{Pitman}), there exist $\tilde{n}=\min\{n,K\}$ independent Bernoulli variables $Y_{i}$, with $i=1,\ldots{}\tilde{n}$, such that $\hat{X}=\sum_{i=1}^{\tilde{n}}Y_{i}$. Clearly then, one can equivalently describe $\hat{X}$ as a sum of exactly $n$ independent Bernoulli variables, $\hat{X}=\sum_{i=1}^{n}Y_{i}$, obtained by padding the sequence $(Y_{1},\ldots,Y_{\tilde{n}})$ with $n-\tilde{n}$ Bernoulli variables deterministically equal to 0. This being the case, $\hat{p}=\hat{X}/n$ fulfills the Bernstein inequality, which for the lower tail yields~\cite{Bennett}
\begin{equation}\label{actual_bernstein}
\Pr\left[\hat{p}\leq{}\frac{1}{n}\sum_{i=1}^{n}p_{i}-\delta\right]\leq{}\exp\left[-\frac{n\delta^{2}}{\displaystyle{\frac{2}{n}\sum_{i=1}^{n}p_{i}(1-p_{i})+\frac{2\delta}{3}}}\right]
\end{equation}
for all $\delta\in[0,p]$, with $p_{i}=\mathbb{E}[Y_{i}]$ and $i=1\ldots{}n$. At this point, Eq.~(\ref{bernstein}) follows replacing the expectation (variance) of $\hat{X}$, when interpreted as a sum of independent Bernoulli variables, by its expectation (variance) when interpreted as a hypergeometric random variable. This amounts to substituting $\sum_{i=1}^{n}p_{i}$ by $np$ in the left-hand side of Eq.~(\ref{actual_bernstein}), and $\sum_{i=1}^{n}p_{i}(1-p_{i})$ by $np(1-p)(1-f_{n})$ ---with $f_{n}=(n-1)/(N-1)$--- in the right-hand side of Eq.~(\ref{actual_bernstein}).

The validity of Eq.~(\ref{bernstein}) without the constraint $n/N\leq{}p$ makes it a suitable tool for statistical inference. To be precise, applying the well-known method of Sec.~\ref{III} and Sec.~\ref{hush_detail}, the following confidence upper bound can be derived from Eq.~(\ref{bernstein}).
\begin{proposition}
Let $\hat{X}\sim{}Hypergeometric(N,K,n)$ and $\hat{p}=\hat{X}/n$, such that $\mathbb{E}\left[\hat{p}\right]=p$ with $p=K/N$. Also, let $\pi_{n,\epsilon}=(2/3n)\ln(1/\epsilon)$ and $f_{n}=(n-1)/(N-1)$. Then, for all $\epsilon>0$, $\Pr[\Gamma^+_{N,n,\epsilon}(\hat{p})\leq{}p]\leq{\epsilon}$, where $\Gamma^+_{N,n,\epsilon}(\hat{p})=\gamma^{+}_{N,n,\epsilon}(\hat{p})$ if $\hat{p}\in\left[0,1-\pi_{n,\epsilon}\right]$ and $\Gamma^+_{N,n,\epsilon}(\hat{p})=1+\epsilon$ otherwise, with
\begin{equation}
\gamma^{+}_{N,n,\epsilon}(x)=\frac{3\pi_{n,\epsilon}\left(1-f_{n}\right)+2x+\pi_{n,\epsilon}+\sqrt{\pi_{n,\epsilon}^{2}+3\pi_{n,\epsilon}\left(1-f_{n}\right)\Bigl[3\pi_{n,\epsilon}\left(1-f_{n}\right)+4x(1-x)+2\pi_{n,\epsilon}(1-2x)\Bigr]}}{2\Bigl[1+3\pi_{n,\epsilon}\left(1-f_{n}\right)\Bigr]}.
\end{equation}
\end{proposition}
Plugging the confidence bound of Proposition 8 in Eq.~(\ref{hush_th}), the Greene \& Wellner threshold function follows, which provides the failure probability estimation of the BBM92 protocol via Proposition 3 in the main text. As for the decoy-state BB84 protocol, all considerations made for the Hush \& Scovel threshold function remain valid for Greene \& Wellner's. In particular, assuming a common error probability $\varepsilon$ per statistical bound, and defining $\phi_{1,\rm Z}^{\rm U}$ via Eq.~(\ref{pher_HS}) ---with the confidence upper bound of Proposition 8---, $\Pr\left[\Omega_{\rm TEST},\Omega_{\rm PE}\right]\leq{}16\varepsilon$ follows for the PE test of Eq.~(\ref{simpler}).

\subsection{QKD simulations}\label{QKD_simul}
Figure~\ref{newlineHS} below reproduces Figure 2 in the main text, but incorporating additional lines relying on Hush \& Scovel and Greene \& Wellner inequalities. Specifically, dashed-dotted blue lines are deployed in both cases, using single dots for the former tool and double dots for the latter. In Figure~\ref{newlineHS}b, the new lines use Proposition 1 in the main text for the decoy-state confidence bounds, matching the criterion of the solid blue line for a fair comparison.

\begin{figure}[htbp!]
\centering
 \begin{subfigure}{.5\textwidth}
   \centering
   \includegraphics[width=\columnwidth]{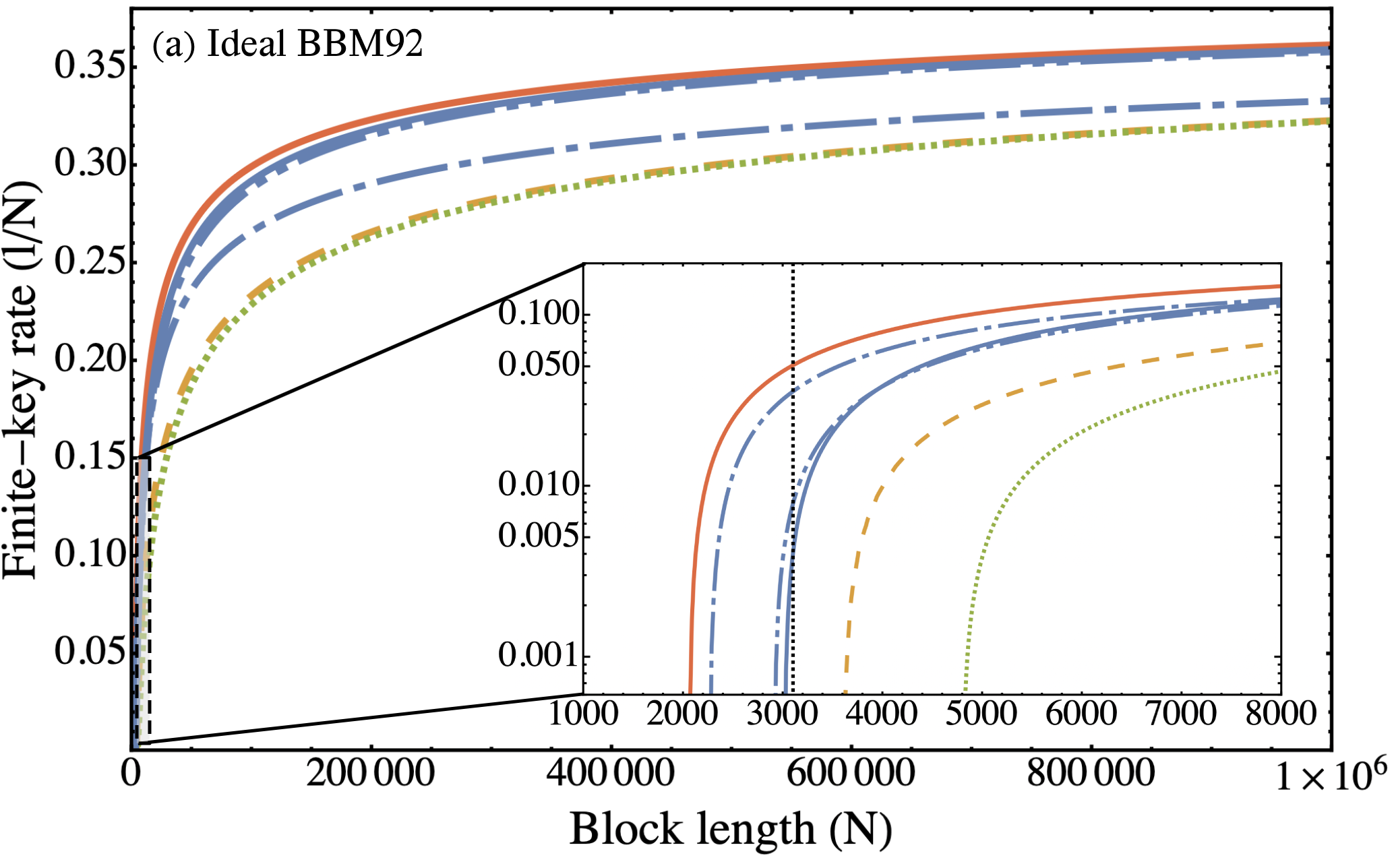}
   \end{subfigure}%
  \hfill
 \begin{subfigure}{.5\textwidth}
   \centering
   \includegraphics[width=\columnwidth]{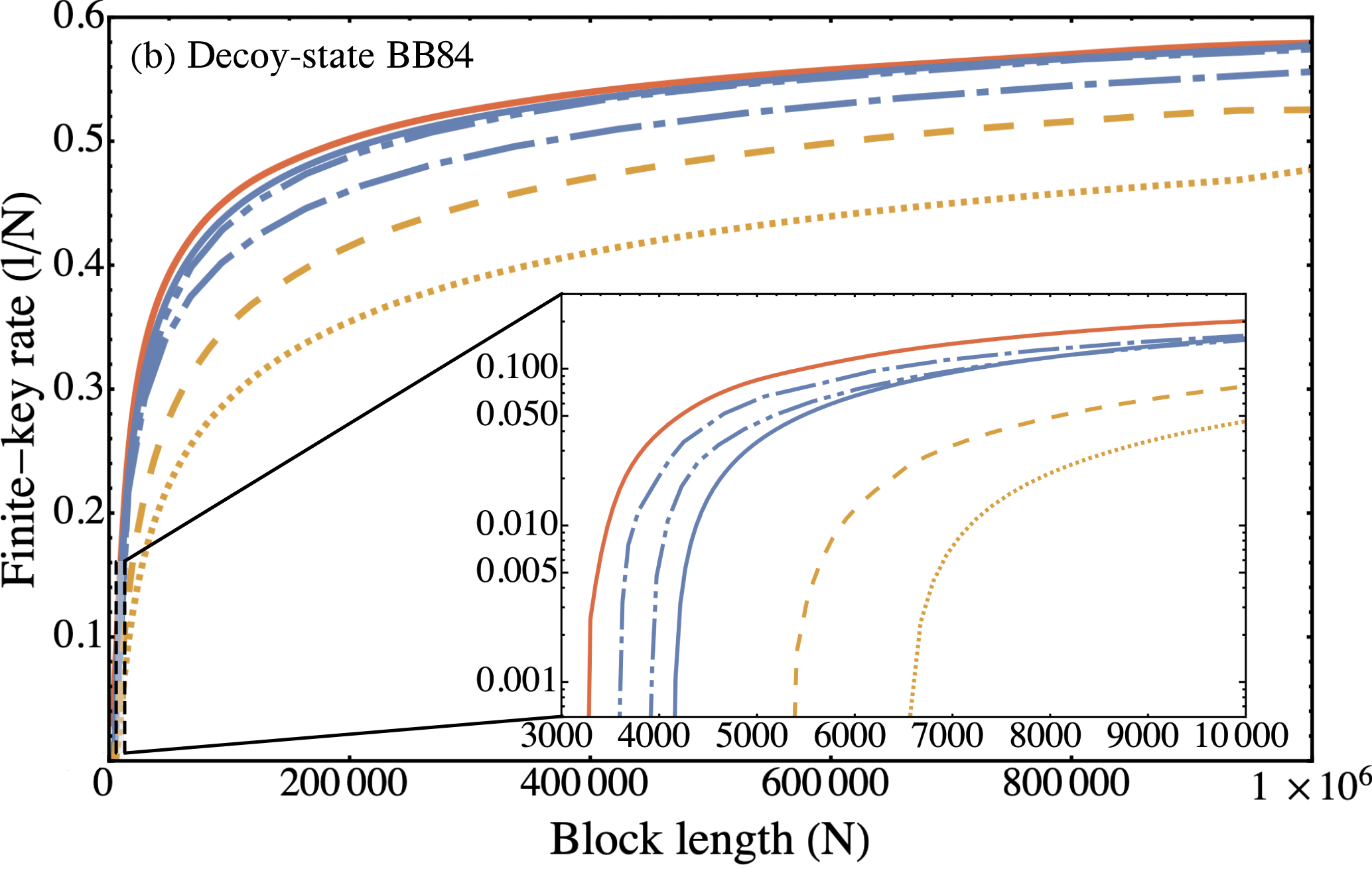}
    \end{subfigure}
    \caption{Finite secret key rate versus data block size. (a) Ideal BBM92 protocol. Dotted green line: Eq.~(\ref{Serfling_threshold}); dashed orange line: Eq.~(\ref{LimEkert}); solid blue line: Eq.~(\ref{threshold}); dashed-dotted blue line, single dots: Eq.~(\ref{hush_th}) via Proposition 7; dashed-dotted blue line, double dots: Eq.~(\ref{hush_th}) via Proposition 8; solid red line: Eq.~(\ref{threshold_CP}). The dotted black vertical line in the inset indicates the block size gathered in the Micius experiment~\cite{Micius}, $N=3100$ bits. (b) Decoy-state BB84 protocol. Dotted orange line: Eq.~(\ref{LimEkert}) plus Eq.~(\ref{hoeffding_pb}); dashed orange line: Eq.~(\ref{LimEkert}) plus Eq.~(\ref{mult_pb}); solid blue line: Eq.~(\ref{threshold}) plus Eq.~(\ref{add_pb}); dashed-dotted blue line, single dots: Eq.~(\ref{hush_th}) via Proposition 7 plus Eq.~(\ref{add_pb}); dashed-dotted blue line, double dots: Eq.~(\ref{hush_th}) via Proposition 8 plus Eq.~(\ref{add_pb}); solid red line: Eq.~(\ref{threshold_CP}) plus Eq.~(\ref{CP_pb}). The considered protocol inputs and experimental settings are specified in the main text.}
    \label{newlineHS}
\end{figure}

According to the insets of the figure, both Hush \& Scovel and Greene \& Wellner inequalities perform better than our relaxed Chernoff bound (solid blue line) if small enough block sizes are considered, as already anticipated in Sec.~\ref{V}. In particular, a substantial advantage is observed for Hush \& Scovel, although always smaller than the one obtained with the numerical methods (solid red line). On the contrary, Hush \& Scovel becomes looser than Greene \& Wellner for moderate-to-large blocks, and our relaxed Chernoff bound performs better than both of them with the considered settings. We note, however, that the advantage with respect to Greene \& Wellner is small in this regime.


\begin{thebibliography}{50}
\bibitem{robustness}
In fact, strictly speaking, the ``robustness" penalty of using these thresholds is critical when considering the more convoluted PE test devised for the random sampling tool of~\cite{Ekert}. This is so because, in that test, double-sided thresholds are imposed, and since both the upper and the lower thresholds are being set to the same value for the simulations, the success probability of the test becomes negligible.
\bibitem{justification}
In this work, we contemplate the equality sign following the criterion of~\cite{Leverrier}. With this criterion, ``0" and ``1" do not provide trivial confidence bounds on the population parameter, forcing us to take the looser values given in Eq.~(\ref{Gamma-}) and Eq.~(\ref{Gamma+}).
\bibitem{comment_2}
For instance, it is straightforward to show that setting $\epsilon>\exp[-\frac{9(N-n)n}{4(2n+N)}]$ suffices to assure that $\Gamma^{+'}_{n,\epsilon}(p^{\rm th})\geq{}n/N$ for all $p^{\rm th}<1/2$. This condition on $\epsilon$ is unrestrictive for all practical purposes.
\bibitem{clarification}
The statistic $\mathcal{CP}^{+}_{N,n,\epsilon}(\hat{p})$ we introduce here does not exactly match ---but is trivially determined by--- the one presented in~\cite{Buonaccorsi,Wright}. Precisely, they are related by a ``one-unit shift". The discrepancy arises because we contemplate the equality sign in the bounds, requesting $\Pr\left[\mathcal{CP}^{+}_{N,n,\epsilon}(\hat{p})\leq{}p\right]\leq{}\epsilon$ rather than $\Pr\left[\mathcal{CP}^{+}_{N,n,\epsilon}(\hat{p})<p\right]\leq{}\epsilon$.
\bibitem{speed-up}
One can speed-up the computation of $\mathcal{CP}^{+}_{N,n,\epsilon}(x)$ by exploring the complementary set of potential $p$'s (namely, $\left\{p\geq{}x\big\rvert{}\Pr\bigl[\hat{p}\leq{}x|p\bigr]>\epsilon\right\}$), by carefully selecting the starting $p$ for the iterations (e.g. based on existing analytical bounds), or by applying a binary search or an interpolation algorithm.
\bibitem{clarification2}
Note that we use different variables for the correctable QBER in the 
BBM92 and the BB84 protocols, respectively given by $p^{\rm th}$ and $\theta^{\rm th}$. For consistency, $\theta^{\rm th}$ is set to the expected QBER of the considered channel model.
\bibitem{clarification3}
The total number of error terms contributing to $\epsilon_{\rm PE}$ varies depending on the random sampling tool under consideration, and the reader is referred to Appendix II for the details.
\bibitem{counterfactual}
A relevant observation is that, in order to enable the counterfactual argument given in~\cite{Lim} (on which the finite statistical analysis is based), we assume that the intensity assignments are performed only after the data sampling.
\bibitem{subtlety}
Note that the equality sign in the probability $\Pr[\Gamma^{-}_{n,\epsilon}(\hat{p})\geq{}p]$ prevents us from defining $\Gamma^{-}_{n,\epsilon}(\hat{p})=0$ directly for $\hat{p}\notin{}I_{n,\epsilon}$. In particular, let $\Upsilon^{-}_{n,\epsilon}(\hat{p})$ follow this alternative definition. Then, if $p=0$, for sure $\hat{p}=0\notin{}I_{n,\epsilon}$ and hence $\Upsilon^{-}_{n,\epsilon}(\hat{p})=0$, such that $\Pr\left[\Upsilon^{-}_{n,\epsilon}(\hat{p})\geq{}p\right]=\Pr\left[\Upsilon^{-}_{n,\epsilon}(\hat{p})=p\right]=1>\epsilon$. On the contrary, $\Upsilon^{-}_{n,\epsilon}(\hat{p})$ would fulfill the relaxed claim contemplating the strict inequality, $\Pr\left[\Upsilon^{-}_{n,\epsilon}(\hat{p})>p\right]\leq{}\epsilon$.
\bibitem{refinement}
 As conjectured by Serfling and proven in~\cite{Bardenet}, the dependence of Serfling inequality on the sampling fraction can be slightly refined using a reverse martingale argument. Nevertheless, the improvement is small and only shows up for $n/N>1/2$, which we exclude in our discussion.
\bibitem{refinement2}
 The threshold function of Eq.~(\ref{Serfling_threshold})---derived in Lemma 6 of~\cite{Leverrier}--- does not follow from the Serfling lower bound on $\hat{p}$ we discuss here, but from the Serfling upper bound on the complementary proportion $\hat{q}$. Interestingly, the function that follows from our approach (which can be obtained from Eq.~(\ref{Serfling_threshold}) exchanging $n$ and $N-n$) is strictly tighter for $n/N<1/2$.
\bibitem{casuistry}
 Note that, with the specifications $p<1/2$ and $n/N<1/2$, demanding $p<\min(n/N,1-n/N)$ or $p>\max(n/N,1-n/N)$ is equivalent to demanding $p<n/N$. Also, the original statement demands $\delta>2/n$, but this constraint can be trivially dropped without compromising the statement.
\bibitem{slope}
The condition on the slope requested in Proposition 3 is again unrestrictive. To be precise, this condition reads $\Gamma^{+'}_{N,n,\epsilon}(p^{\rm th})\geq{}n/N$, which is trivially guaranteed for all $p^{\rm th}<1/2$ if one sets $\epsilon>\exp[-\frac{2(N+2)(N-n)n}{N^{2}}]$.
\end{thebibliography}
\end{document}